\documentclass[preprint,12pt,3p]{elsarticle}
\usepackage{amssymb,amsmath,amsthm,enumerate,latexsym}
\usepackage{graphicx}

\newtheorem{theorem}{Theorem}
\newtheorem{lemma}[theorem]{Lemma}
\newtheorem{corollary}[theorem]{Corollary}
\newdefinition{definition}[theorem]{Definition}
\newdefinition{question}[theorem]{Open Question}
\newdefinition{remark}[theorem]{Remark}
\newdefinition{example}[theorem]{Example}

\newcommand{\F}{F_{\mathcal{A}}}
\newcommand{\FB}{F_{\mathcal{B}}}
\newcommand{\fa}{f_{\mathcal{A}}}
\newcommand{\fb}{f_{\mathcal{B}}}

\newcommand{\fai}{f_{\mathcal{A}_i}}
\newcommand{\faij}{f_{\mathcal{A}_{ij}}}
\newcommand{\faj}{f_{\mathcal{A}_j}}
\newcommand{\Fa}{F_{\mathcal{A}}}
\newcommand{\Fav}{F_{\mathcal{A}_v}}
\newcommand{\Fai}{F_{\mathcal{A}_{i}}}
\newcommand{\Faij}{F_{\mathcal{A}_{ij}}}

\newcommand{\Fap}{F_{\mathcal{A}'}}
\newcommand{\fh}{\hat{f}}
\newcommand{\gh}{\hat{g}}
\newcommand{\fah}{\hat{f}_{\mathcal{A}}}
\newcommand{\faijh}{\hat{f}_{\mathcal{A}_{ij}}}

\newcommand{\fbh}{\hat{f}_{\mathcal{B}}}
\newcommand{\A}{\mathcal{A}}
\newcommand{\B}{\mathcal{B}}
\newcommand{\C}{\mathcal{C}}
\newcommand{\D}{\mathcal{D}}
\newcommand{\So}{\Sigma^{\omega}}
\newcommand{\Ss}{\Sigma^ *}
\newcommand{\Sow}{{\Sigma^{\omega}\setminus\Ss 1^{\omega}}}

\newcommand{\ds}{\displaystyle}

\let\epsilon\varepsilon

\def\picture#1#2#3{\begin{figure}
\begin{center}
\includegraphics{#2}
\caption{#1}
\def\tmp{#3}
\ifx\tmp\empty\else\label{#3}\fi
\end{center}
\end{figure}}

\journal{Linear Algebra and its Applications}

\begin{document}

\begin{frontmatter}

\title{On Continuous Weighted Finite Automata\tnoteref{abbr}}

\tnotetext[abbr]{Abbreviations used: ap (average preserving), RCP (right-convergent product), WFA (weighted
finite automata)}

\author[utu]{Jarkko Kari}
\ead{jkari@utu.fi}
\author[uk]{Alexandr Kazda\corref{cor}}
\ead{alexak@atrey.karlin.mff.cuni.cz}
\author[utu]{Paula Steinby}

\address[utu]{Department of Mathematics, University of Turku\\ Vesilinnantie 5, FI-20014, Turku, Finland}
\address[uk]{Department of Algebra, Charles University\\ Sokolovsk\'a 83, 186 75, Praha 8, Czech Republic}
\cortext[cor]{Corresponding author}

\begin{abstract}
We investigate the continuity of the $\omega$-functions and real functions
defined by weighted finite automata (WFA). We concentrate on the case of average
preserving WFA. We show that every continuous
$\omega$-function definable by some WFA can be defined by an average preserving
WFA and then characterize minimal average preserving WFA whose $\omega$-function 
or $\omega$-function and real function are continuous. 

We obtain several algorithmic reductions for WFA-related decision problems. In
particular, we show that deciding whether the $\omega$-function and real
function of an average preserving WFA are both continuous is computationally
equivalent to deciding stability of a set of matrices.

We also present a method for constructing WFA that compute continuous real
functions.
\end{abstract}


\end{frontmatter}

\section{Introduction}
Weighted finite automata (WFA) over $\mathbb R$ are finite
automata with transitions labelled by real numbers. They can be viewed
as devices to compute functions from words to real numbers, or
even as a way to define real functions. Weighted finite automata and transducers
have many nice applications in natural language processing, image manipulation
etc, see~\cite{handbook,Hafner98weightedfinite,Kari,Kari2,tarragonabook,mohri1,mohri2} and references therein.
On the other hand, weighted automata over more general semi-rings have been
extensively studied as a natural extension of ordinary automata. A good source for more background
information on various aspects of weighted automata is
a forthcoming handbook on the field~\cite{handbook}.

WFA provide a natural and intrinsic description for some
self-similar real functions.
Smooth real functions defined by WFA are limited to
polynomials~\cite{Karh3,Droste}. However, many more functions in
lower differentiability classes can be generated. In this paper,
we study those WFA functions that are continuous. We have two
concepts of continuity: the continuity of the function $\fa$ that
assigns real numbers to infinite words, and the continuity of the corresponding
function $\fah$ that assigns values to points of the unit interval.

Culik and Karhum\"aki have stated various results about real functions defined
by the so called \emph{level automata} in \cite{Karh}. In this work, we generalize
many of these results to apply to the general setting. We also use the closely related
theory of  right-convergent product (RCP)
sets of matrices as developed in \cite{Daubechies,Daubechies2} and \cite{Wang}.

The paper is organized as follows: In Section~\ref{preliminaries} we give basic
definitions, introduce the concepts of stable sets and RCP sets
of matrices, provide some key results on RCP sets from~\cite{Daubechies,Daubechies2},
define weighted finite automata and discuss the important concepts of minimality
and average preservation (ap).

In Section~\ref{omegasection} we study WFA as devices that assign real numbers to
infinite words. We prove that any continuous function that can be defined by a WFA can, in fact,
be defined using average preserving WFA, so restricting the attention to ap WFA is
well motivated. We establish a canonical form for the average preserving
WFA whose $\omega$-function $\fa$ is continuous (Corollary~\ref{matrix-reloaded}).
We obtain algorithmic reductions between the decision problems of determining
convergence and continuity of $\fa$, and the stability, product convergence and
continuous product convergence
of matrix sets.  If stability of finite sets of square matrices is undecidable (which
is not presently known) then all questions considered are undecidable as well.

In Section~\ref{realfunctionsection} we consider the real functions defined by WFA.
Connections between the continuity of the $\omega$-function and the corresponding real function
are formulated. We specifically look into those ap WFA whose $\omega$- and real functions are both
continuous. If the $\omega$-function is continuous then there is a simple and effectively testable additional
condition for the continuity of the corresponding
real function. Again, we see that the stability of matrix products plays an
important role in algorithmic questions. Finally, we provide a method to generate continuous
ap WFA when a stable pair of matrices is given.

\section{Preliminaries}
\label{preliminaries}
Let $\Sigma$ be a non-empty finite set. In this context, we call
$\Sigma$ an \emph{alphabet} and its elements \emph{letters}. With
concatenation as the binary operation and the \emph{empty word}
$\epsilon$ as the unit element, $\Sigma$ generates the monoid $\Ss$,
the elements of which are called \emph{words}.

We denote by $|v|$ the \emph{length} of the word $v\in\Ss$. Denote the $i$-th
letter of the word $v$ by $v_i$ and the factor $v_iv_{i+1}\cdots v_j$ by
$v_{[i,j]}$. By $pref_k(v)$ we denote the prefix of length $k$ of the word $v$. An
\emph{infinite word} $w$ is formally a mapping $w:\mathbb{N}\rightarrow$
${\Sigma}$. Denote the set of all infinite words by $\So$.

The set $\So$ is a metric space with the Cantor metric (or prefix metric)
defined as follows:
\[
d_C(w,w')=\Big\{
\begin{array}{ll}
    0 & \textrm{ if }w=w',\\
    \frac{1}{2^k}&\textrm{ otherwise,}
\end{array}
\]
where $k$ is the length of the longest common prefix of $w$ and $w'$. The space
$\So$ is a product of the compact spaces $\Sigma$, therefore $\So$ itself is
compact.

The set of reals $\mathbb{R}$ is a complete metric space with the usual Euclidean metric
\[d_E(x,y)=|x-y|\textrm{ for all }x,y\in\mathbb{R}.\]

We denote by $E$ the unit matrix (of appropriate size). We use the
same notation $\|A\|$ both for the usual $l^2$ vector norm, if $A$
is a vector, and for the corresponding matrix norm, if $A$ is a
matrix.

Assume that for each letter $a\in\Sigma$
we have an $n\times n$ square matrix $A_a$. Then for $v\in\Ss$ let $A_v$ denote the matrix product
$A_{v_1}A_{v_2}\ldots A_{v_k}$. If $v$ is empty, let $A_v=E$. If $w\in\So$,
we let $A_w=\ds\lim_{k\to\infty}A_{pref_k(w)}$ if the limit exists.

In this paper, we assume that the elements of all matrices are defined in such
a way that we can algorithmically perform precise operations of addition,
multiplication and division as well as decide equality of two numbers. We can obtain such effective
arithmetics by limiting ourselves, for example, to matrices and vectors with
rational elements.

\begin{definition} Let $B=\{A_a\,|\,a\in\Sigma\}$ be a nonempty set of $n\times n$ matrices
such that $A_w=0$ for every $w\in\So$. Then we call $B$ a \emph{stable set}.
\end{definition}

Given a finite set $B$ of matrices, we will call the algorithmic
question ``Is $B$ stable?'' the \emph{Matrix Product Stability}
problem. For $|B|=1$, Matrix Product Stability is easy to solve
using eigenvalues and Lyapunov equation (see \cite[page
169]{Mahmoud}). Moreover, there is a semi-algorithm that halts iff
$B$ is a stable set (idea of this algorithm is to check whether the
joint spectral radius of $B$ is less than 1, see \cite{Daubechies}).
However, it is not known whether there exists an algorithm deciding
Matrix Product Stability; even the binary ($|B|=2$) case is as hard
as the general stability problem, see \cite{unsolvedblondel} and
\cite{Blondel2} (we also prove this in
Lemma~\ref{reductiontosizetwo}).

The following Lemma is stated as Corollary 4.1a in \cite{Daubechies}. For the sake
of completeness we offer
a short proof here.

\begin{lemma}\label{uniform}
Let $\{A_a\,|\,a\in\Sigma\}$ be a stable set of $n\times n$ matrices.
Then the convergence of $A_{pref_k(w)}$ to zero is uniform. That is,
for every $\epsilon>0$ there exists $k_0$ such that for any
$v\in\Ss$ with $|v| > k_0$ we have $\|A_v\|<\epsilon$.
\end{lemma}
\begin{proof}
Assume that the statement is not true.
Then there is an $\epsilon>0$ such that there exist arbitrarily long $v\in\Ss$
such that
$\|A_v\|\geq \epsilon$. From compactness of $\So$ we obtain that
there exists an infinite word $u\in \So$ with the property that for
each $l$ there exists a $v_l\in \Ss$ such that
$\|A_{u_{[1,l]}v_l}\|\geq \epsilon$. But then
\begin{align*}
\|A_{u_{[1,l]}}\|\|A_{v_l}\|&\geq\|A_{u_{[1,l]}}A_{v_l}\|\geq \epsilon, \mbox{ so} \\
\|A_{v_l}\|&\geq \frac\epsilon{\|A_{u_{[1,l]}}\|}.
\end{align*}

Because $A_u=0$, we have that the set $\{\|A_v\|\,|\,v\in\Ss\}$ is unbounded. In
the rest of the proof, we use the reasoning from \cite{Gurvits}
(proof of Lemma~1.1).

For each $k$, let $v^{(k)}$ be a  word of length at most $k$ such
that $\|A_{v^{(k)}}\|$ is maximal (note that $v^{(k)}$ might be
empty). Denote $l=|v^{(k)}|$. First, we show that
$\|A_{v^{(k)}_{[1,i]}}\|\geq1$ for all $1\leq i\leq l$. If for some
such $i$ we would have $\|A_{v^{(k)}_{[1,i]}}\|<1$, then
\begin{align*}
\|A_{v^{(k)}_{[1,i]}}\|\|A_{v^{(k)}_{[i+1,l]}}\|&\geq \|A_{v^{(k)}}\|, \mbox{ so}\\
\|A_{v^{(k)}_{[i+1,l]}}\|&>\|A_{v^{(k)}}\|,
\end{align*}
contradicting the maximality of $\|A_{v^{(k)}}\|$. We conclude that
$\|A_{v^{(k)}_{{[1,i]}}}\|\geq 1$ for all $i$.

As the norm of matrices $A_v$ is unbounded, the length of $v^{(k)}$
goes to infinity. Then we obtain from the compactness of $\So$ that
there exists a word $w\in\So$ such that for each $i$ we can find
$k_i$ such that $w_{[1,i]}=v^{(k_i)}_{[1,i]}$. But this means that
$\|A_{w_{[1,i]}}\|\geq 1$ for each $i$, a contradiction with the
stability of $\{A_a\,|\,a\in\Sigma\}$.
\end{proof}

\begin{definition}
A set of matrices $\{A_a\,|\,a\in\Sigma\}$ is called
\emph{right-convergent product set} or RCP set if the function
$A_w:w\mapsto \ds\lim_{k\to \infty} A_{pref_k(w)}$ is defined on the
whole set $\So$. If $A_w$ is continuous on $\So$, we say that the
set is \emph{continuous RCP}.
\end{definition}

Clearly every stable set is a continuous RCP set.
In \cite{Daubechies}, the authors prove several results about RCP sets of
matrices. Most importantly, Theorem 4.2 from \cite{Daubechies} (with errata
from \cite{Daubechies2}) gives us a
characterization of continuous RCP sets of matrices. For $V,E_1$ subspaces of
${\mathbb R}^n$ such that $R^n=V\oplus E_1$ denote by $P_V:{\mathbb R}^n\to
{\mathbb R}^n$ the projection to $V$ along $E_1$.

\begin{theorem}[Theorem 4.2 in \cite{Daubechies}]\label{Daubechies}
Let $B=\{A_a\,|\,a\in\Sigma\}$ be a finite set of $n\times n$ matrices. Then the following
conditions are equivalent:
\begin{enumerate}[(1)]
\item The set $B$ is a continuous RCP set.
\item All matrices $A_a$ in $B$ have the same left 1-eigenspace $E_1 =
E_1(A_a)$, and this eigenspace is simple for all $A_a$. There exists
a vector space $V$ with ${\mathbb R}^n = E_1 \oplus V$, having the
property that $P_{V} B P_V$ is a stable set.
\item The same as (2), except that $P_V B P_V$ is a stable set
 for all vector spaces $V$ such that ${\mathbb R}^n = E_1 \oplus V$.
\end{enumerate}
\end{theorem}
From Theorem~\ref{Daubechies} and Lemma~\ref{uniform} follows a corollary
(stated as Corollary 4.2a in \cite{Daubechies}), which generalizes Lemma~\ref{uniform}:
\begin{corollary}\label{cont-conv}
If $B=\{A_a\,|\,a\in\Sigma\}$ is a continuous RCP set then all the products
$A_{pref_k(w)}$ for $w\in\So$ converge uniformly at a geometric rate.
\end{corollary}

Another important result that we will need is part (a) of Theorem I
in \cite{Wang}, stated below (slightly modified) as Theorem~\ref{wang2}.
A set of matrices $B$ is called
\emph{product-bounded} if there exists a constant $K$ such that the
norms of all finite products of matrices from $B$ are less than $K$.
Notice that as all matrix norms on $n\times n$ matrices are
equivalent, being product-bounded does not depend on our choice of
matrix norm.

\begin{theorem}\label{wang2}
Let $B$ be an RCP set of matrices. Then $B$ is
product-bounded.
\end{theorem}

So we have the following sequence of implications:
$$
\mbox{$B$ stable } \Longrightarrow \mbox{ $B$ continuous RCP }
\Longrightarrow \mbox{ $B$ RCP } \Longrightarrow \mbox{ $B$ product
bounded}
$$
The problem of determining whether a given finite $B$ is product
bounded is undecidable~\cite{Blondel2}, while it is not known
whether it is decidable if $B$ is stable, RCP or continuous RCP. We
will explore the relationship between RCP sets and WFA later in the
paper.

\subsection*{Weighted Finite Automata}
A \emph{weighted finite automaton} (WFA) $\A$ is a
quintuple $(Q,\Sigma,I,F,\delta)$. Here $Q$ is the state set,
$\Sigma$ a finite alphabet, $I:Q\rightarrow\mathbb{R}$ and
$F:Q\rightarrow\mathbb{R}$ are the \emph{initial} and \emph{final
distributions}, respectively, and $\delta:Q\times\Sigma\times
Q\rightarrow \mathbb{R}$ is the \emph{weight function}. If
$\delta(p,a,q)\neq 0$ for $a\in\Sigma$, $p,q\in Q$, we say that there
is a transition from $p$ to $q$ labelled by $a$ of weight $\delta(p,a,q)$.
We denote the cardinality of the state set by $|Q|=n$.
Note that we allow $Q$ to be empty.

A more convenient representation of $\A$ is by vectors
$I\in\mathbb{R}^{1\times n}$, $F\in\mathbb{R}^{n\times 1}$ and a
collection of \emph{weight matrices} $A_a\in \ds\mathcal{M}_{n\times
n}(\mathbb{R})$ defined by
\[\forall a\in\Sigma,\, \forall i,j\in Q:\; (A_a)_{ij}= \delta(i,a,j).\]
A WFA $\A$ defines the \emph{word function} $\Fa:\Ss\rightarrow \mathbb{R}$
by
\[
\Fa (v)=I A_v F.
\]

We denote by $\A_v$ the automaton which we get from the WFA $\A$ by
substituting $IA_v$ for the original initial distribution $I$.
Obviously, $\Fav(w)=\Fa(vw)$ for all $w\in\Ss$.

\begin{remark}
Notice that for any $n\times n$ regular matrix $M$, we can take a new automaton with
distributions $IM,M^{-1}F$ and the set of weight matrices $\{M^{-1}A_aM\,|\,a\in\Sigma\}$
without affecting the computed word function. We will call this operation
changing the basis.

This means that whenever $I,F\neq 0$, we can change either $I$ or $F$ to any
nonzero vector of our choice by switching to a different basis.
\end{remark}

Given a word function $F$, we can define \emph{$\omega$-function} $f$
on infinite words. For $w\in\So$, we let
\begin{equation} \label{limit}
f(w)=\lim_{k\rightarrow\infty}F(pref_k(w)),
\end{equation}
if the limit exists. If the limit does not exist then $f(w)$ remains undefined.
In the following, we will use this construction to define $\omega$-function
$\fa$ using $\Fa$ for some $\A$ weighted finite automaton.

As usual, the $\omega$-function $f:\So\rightarrow\mathbb{R}$ is
\emph{continuous} at $w\in\So$
if for every positive real number $\varepsilon$ there exists a
positive real number $\delta$ such that all $w'$ in $\So$
such that $d_C(w,w') < \delta$ satisfy $d_E(f(w),f(w')) < \varepsilon$.
In particular, if $f$ is continuous at $w$ then
$f$ must be defined in some neighborhood of $w$. We say that $f$ is continuous
if it is continuous at every $w\in\So$.

Thorough the paper, we will be mostly talking about the case when
the convergence in the limit (\ref{limit}) is uniform:

\begin{definition}\label{uc}
We say that a word function $F$ is \emph{uniformly convergent} if
for every $\epsilon>0$, there exists a $k_0$ such that for all $w\in\So$ and
all $k>k_0$ we have
\[
|F(pref_k(w))-f(w)|<\epsilon.
\]
A WFA $\A$ is \emph{uniformly convergent}
if $\F$ is uniformly convergent.
\end{definition}
\begin{lemma}\label{remUCC}
If a word function $F$ is uniformly convergent then the corresponding  $\omega$-function $f$ is defined and
continuous in the whole $\So$.
\end{lemma}
\begin{proof}
From the definition of uniform convergence we obtain that
 $f(w)$ must exist for every $w\in\So$.
Continuity follows from the fact that $f$ is
the uniform limit of continuous functions $f_k$ defined as
$f_k(w)=F(pref_k(w))$ for all $w\in\So$.
\end{proof}

The following Lemma gives another formulation of the uniform convergence.

\begin{lemma}\label{lemStrongerUC}
The function $F$ is uniformly convergent iff for each $\epsilon>0$ there exists
$m$ such that for all $w\in\So$ and $v\in\Ss$ such that
$pref_m(v)=pref_m(w)$ we have 
\[
|F(v)-f(w)|<\epsilon.
\]
\end{lemma}
\begin{proof}
Obviously, if $F$ satisfies the condition on the right then
letting $k_0=m$ and $v=pref_k(w)$ for $k>k_0$ yields that $F$ is uniformly convergent.

For the converse, assume $\epsilon>0$ is given. We need to
 find $m$ with the required properties.

If $F$ is uniformly convergent then $f$ is
continuous by Lemma~\ref{remUCC} and from the compactness of $\So$ we obtain
uniform continuity of $f$. So there exists $l$ such that $pref_l(w)=pref_l(z)$ implies
$|f(w)-f(z)|<\epsilon/2$ for $w,z\in \So$.  Let now $k_0$ be such that
$|F(pref_k(z))-f(z)|<\epsilon/2$ for all $z$ and all $k>k_0$. Choose
$m>k_0,\,l$. Given $v\in\Ss, w\in\So$ with $pref_m(v)=pref_m(w)$, choose
$z\in\So$ such that $v$ is a prefix of $z$. Then we can write:
\[
|F(v)-f(w)|\leq
|F(v)-f(z)|+|f(z)-f(w)|<\frac12\epsilon+\frac12\epsilon=\epsilon,
\]
concluding the proof.
\end{proof}

In contrast to Lemma~\ref{uniform},
the following example shows that convergence to zero everywhere does not guarantee that a WFA
converges uniformly.

\begin{example}\label{universal-counterexample}
Consider the automaton $\A$ on the alphabet $\Sigma=\{0,1\}$ described by Figure~\ref{u-c-figure}.
\picture{The automaton from Example~\ref{universal-counterexample}}{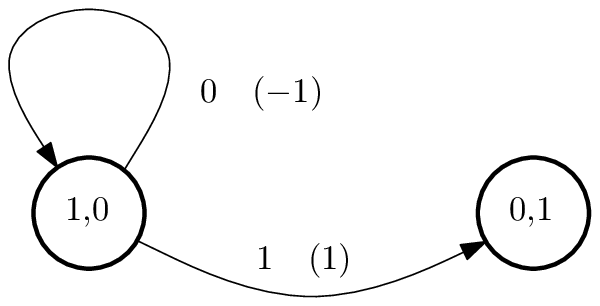}{u-c-figure}
In the figure, the two numbers inside each state denote the initial and final distribution, respectively,
while the numbers next to the arrow express the label and weight of the transition
(weight is in parentheses).
The matrix presentation of this automaton is
$$
I=
\begin{pmatrix}
1 & 0
\end{pmatrix}
\quad
A_0=
\begin{pmatrix}
-1 & 0\\ 0 & 0
\end{pmatrix}
\quad
A_1=
\begin{pmatrix}
0 & 1\\ 0 & 0
\end{pmatrix}
\quad
F=
\begin{pmatrix}
0 \\ 1
\end{pmatrix}
$$
This automaton computes the word function
\[
\F(v)=\Big\{
\begin{array}{ll}
    (-1)^n & v\in 0^n1 \\
    0 & \textrm{otherwise},\\
\end{array}
\]
so it defines the zero $\omega$-function. However, the convergence is not uniform at
the point $w=0^{\omega}$.
\end{example}

\begin{definition} A WFA $\A$ with $n$ states is said to be \emph{left minimal} if
\begin{equation}\label{leftmin}
\dim\langle IA_u,\; u\in\Ss\rangle=n.
\end{equation}
Similarly, $\A$ is called \emph{right minimal} if
\begin{equation}\label{rightmin}
\dim\langle A_uF,\; u\in\Ss\rangle=n.
\end{equation}
If $\A$ is both left and right minimal, we call it \emph{minimal}.
\end{definition}
In other words, $\A$ is minimal when each of its distributions
generates the space $\mathbb{R}^n$. Moreover, $\A$ is minimal
according to our definition iff it is also minimal in the sense that
no other WFA with fewer states than $n$ can compute
the same word function $\Fa$ (see \cite{Droste}, Proposition 3.1).
Observe that minimality is clearly invariant under the change of basis.

\begin{lemma}\label{find-minimal}
Given a WFA $\A$, we can effectively find WFA $\A'$ such that $\A'$
is minimal and $\Fap=\Fa$.
\end{lemma}

For proof, see \cite{Droste}, Proposition 3.1. Also, if the transition matrices of $\A$ had
rational entries then we can choose $\A'$ so that its transition matrices have
rational entries. In the following, we will often assume that $\A$
is minimal.

\begin{definition} A function $F:\Sigma^*\to\mathbb R$ is \emph{average preserving} (ap), if for all $v\in\Sigma^*$,
 \[\sum _{a\in\Sigma}F(va)=k F(v),\textrm{ where }k=|\Sigma|.\]
The WFA $\A$ with the final distribution $F$ and weight matrices $A_a$
is called \emph{average preserving} (ap), if
\[\sum _{a\in\Sigma} A_a F = k F,\textrm{ where }k=|\Sigma|.\]
\end{definition}

Every ap WFA defines an average preserving word function and
every average preserving word function definable by some WFA can be
defined by an ap WFA (see \cite[pages 306, 310]{Kari}). In fact, any minimal
WFA computing an ap word function must be ap. Notice also that neither a change of
basis nor minimizing (as in Lemma~\ref{find-minimal}) destroys the ap property
of an ap automaton.

\begin{lemma}
\label{zero-ap-lemma}
The only ap word function defining the zero $\omega$-function
is the zero function.\label{ap-zero}
\end{lemma}
\begin{proof}
Assume  that for every $w\in\So$ we have $f(w)=0$,
yet (without loss of generality) $F(v)=s>0$ for some word $v\in\Ss$. By the ap property of $F$, we have:
\[
\frac{1}{|\Sigma|}\sum_{a\in\Sigma}F(va)=F(v).
\]
This means that $\max\{F(va)\,|\,a\in\Sigma\}\geq F(v)$ and so $F(va)\geq s$ for
some $a$. Repeating this argument, we obtain $w\in\So$ such that
$\fa(vw)\geq s>0$, a contradiction.
\end{proof}

\begin{corollary}\label{there-can-be-only-one}
Let $F,G$ be two ap functions defining the same $\omega$-function $f$ and suppose that
$f(w)$ is defined for every $w\in \So$. Then $F=G$.
\end{corollary}
\begin{proof}
As the function $F-G$ is ap and defines the zero $\omega$-function, 
Lemma~\ref{ap-zero} gives us that $F-G=0$ and so $F=G$.
\end{proof}

\begin{example}
As the only ap word function defining $\fa=0$ is the zero function (Lemma
\ref{ap-zero}), it is easy to decide if a given minimal ap WFA $\A$ computes
$\fa=0$. Such automaton is the unique zero state WFA (whose $I,F$ are zero
vectors from the space ${\mathbb R}^0$).
However, in the non-ap case, we might encounter automata such as in Figure
\ref{non-ap}: The automaton
$\A$ has $I=F=(1)$ and $A_0=A_1=\begin{pmatrix}\frac12\end{pmatrix}$. Obviously, $\A$ is
minimal (but not ap) and computes the word function $\F(v)=1/{2^{|v|}}$
and the $\omega$-function $\fa=0$.

\picture{Non ap automaton defining the zero $\omega$-function}{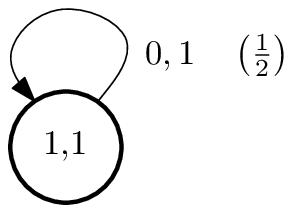}{non-ap}

\end{example}
\section{Properties of $\omega$-functions}
\label{omegasection}
In this section, we will study the $\omega$-function $\fa$, where $\A$ is an
automaton operating on the alphabet $\Sigma$. We will put the emphasis on
$\omega$-functions defined using ap word functions, as these
have numerous useful traits.

\subsection{Average preserving word functions}
We begin by showing that any automaton computing a continuous $\omega$-function can be
modified to be ap and still compute the same $\omega$-function. Hence we do not miss any
WFA definable continuous functions if we restrict the attention to ap WFA.

Actually, the following theorem is
even more general and allows for some ``defects'' of continuity: we only expect $\fa$
to be uniformly continuous on a certain dense set $\Delta\subseteq\So$ (as usual, this means
that the function $\fa$ may even be undefined outside $\Delta$).

We will need the more general formulation later in Section~\ref{realfunctionsection} for
Corollary~\ref{realuniformcorollary}.

\begin{theorem}\label{redistribution} Let $\A$ be a WFA and $w\in\So$.
Suppose that $\fa$ is
continuous on the set $\Delta=\Ss w$, and suppose that there exists a continuous
$g:\So\to \mathbb{R}$ such that ${\fa}_{|\Delta}=g_{|\Delta}$.
Then there is an average preserving WFA $\B$ such that $\fb=g$.
Moreover, if $\A$ is left-minimal then $\B$ can be obtained from $\A$ by changing the final distribution.
\end{theorem}

\begin{remark}
The condition that
``There exists a continuous $g:\So\to{\mathbb R}$ such that
${\fa}_{|\Delta}=g_{|\Delta}$.'' is equivalent with demanding that
${\fa}_{|\Delta}$ be uniformly continuous. Moreover, if such a $g$ exists then
it is unique because $\Delta$ is dense in $\So$.
\end{remark}

\begin{proof}
Using Lemma~\ref{find-minimal}, we can assume
that the automaton $\A=(I,\{A_a\,|\,a\in A\},F)$ is a left-minimal WFA. Denote by $\A_i$
the  WFA obtained from $\A$ by replacing $I$ by $I_i$, the $i$-th
element of the canonical basis of $\mathbb{R}^n$. Let us first
prove that each $\fai$ is uniformly continuous on $\Delta$. From
left-minimality of $\A$ we obtain that there are words
$u_1,\dots,u_n\in\Ss$ and coefficients
$\alpha_1,\dots,\alpha_n\in\mathbb{R}$ such that
$\Fai(v)=\sum_{j=1}^n\alpha_j\Fa(u_jv)$ for all $v\in\Ss$. This
implies that $\fai(v)=\sum_{j=1}^n\alpha_j\fa(u_jv)$ for all
$v\in\Delta$. But then $\fai$ is uniformly continuous on $\Delta$
as a linear combination of uniformly continuous functions
$v\mapsto \fa(u_jv)$.

Recall that $w$ is the infinite word such that $\Delta=\Ss w$. Observe
that the limit $G=\ds\lim_{k\rightarrow\infty}(A_{pref_k(w)}F)$
exists, as we have a simple
formula for the $i$-th component of $G$:
\[G_i=\lim_{k\rightarrow\infty}(I_iA_{pref_k(w)}F)=\fai(w).\]

Denote $L=|\Sigma|$ and let $B=(I,\{A_a\},F')$ be the WFA with the modified final distribution
\begin{equation}\label{fpilkku}
F'=\lim_{i\rightarrow\infty}\left(\frac{\sum_{a\in\Sigma}A_a}L\right)^iG
=\lim_{i\to\infty}\frac{1}{L^i}\sum_{|u|=i}\lim_{k\rightarrow\infty}
A_{pref_k(uw)}F.
\end{equation}
First we show that the limit (\ref{fpilkku}) exists.
The $j$-th coordinate of the $i$-th vector in the sequence
has the following presentation:
\[\phi^{(j)}_i =\frac{1}{L^i}\sum_{|u|=i}\lim_{k\rightarrow\infty}(I_jA_{pref_k(uw)}F)=\frac{1}{L^i}\sum_{|u|=i}\faj (uw).\]
To show that $\ds\lim_{i\rightarrow\infty}\phi^{(j)}_i$ exists, it suffices to show that
$\{\phi^{(j)}_i\}_{i=1}^{\infty}$ is a Cauchy sequence as $\mathbb{R}$ is a
complete metric space.

Let $\varepsilon >0$. As $\faj$ is uniformly continuous on $\Delta$, there is a $k_{\varepsilon}$ such that
$pref_{k_{\varepsilon}}(z)=pref_{k_{\varepsilon}}(z')$ implies
$|\faj(z)-\faj(z')|<\varepsilon$ for any $z,z'\in\Delta$.
Let $\phi^{(j)}_s$ and $\phi^{(j)}_{s+t}$ be two elements of the sequence with
$s\geq k_\epsilon$, $t\in\mathbb{N}$.
Then
\begin{align*}
\left|\phi^{(j)}_{s+t}-\phi^{(j)}_s\right| & = \left|
\ds\frac{1}{L^{s+t}}\ds\sum_{|u|=s+t}\faj (uw)-\frac{1}{L^s}\ds\sum_{|u|=s}\faj
(uw) \right|\\
& = \ds\frac{1}{L^{s+t}}\left|\ds\sum_{|u|=s}\left(\ds\sum_{|v|=t}\faj
(uvw)-L^t\faj(uw)\right) \right|\\
& \leq \ds\frac{1}{L^{s+t}}\ds\sum_{|u|=s}\ds\sum_{|v|=t}\left|\faj
(uvw)-\faj(uw)\right| \\
& < \ds\frac{1}{L^{s+t}}L^{s+t}\varepsilon  = \varepsilon.
\end{align*}

We see that the vector sequence $\{\phi_i\}_{i=1}^\infty$ converges
element-wise, and hence the limit (\ref{fpilkku}) exists.

It remains to show that $\B$ is average preserving and verify the equality
$\fb=g$. To prove the ap property of $\B$, we compute the product
\begin{align*}
\left(\sum_{a\in\Sigma}A_a\right)F' & = \left(\sum_{a\in\Sigma}A_a\right)
\lim_{i\rightarrow\infty}\left(\frac{\sum_{a\in\Sigma}A_a}L\right)^iG\\
& =  L\lim_{i\rightarrow\infty}\frac{\sum_{a\in\Sigma}A_a}L\left(\frac{\sum_{a\in\Sigma}A_a}L\right)^{i}G\\
& =  L\lim_{i\rightarrow\infty}\left(\frac{\sum_{a\in\Sigma}A_a}L\right)^{i+1}G\\
& =  LF'.
\end{align*}

To show that $\fb=g$, let $v\in\So$ be an arbitrary word. Then
\begin{align*}
 \fb(v) & = \lim_{j\rightarrow\infty}IA_{pref_j(v)}F' \\
 & =
 \lim_{j\rightarrow\infty}IA_{pref_j(v)}\lim_{i\rightarrow\infty}\frac{1}{L^i}\sum_{|u|=i}\lim_{k\rightarrow\infty} A_{pref_k(uw)}F \\
 & =  \lim_{j\rightarrow\infty}\lim_{i\rightarrow\infty}
  \sum_{|u|=i} \frac{1}{L^i}\lim_{k\rightarrow\infty} IA_{pref_j(v)pref_k(uw)}F
  \\
 & =  \lim_{j\rightarrow\infty}\lim_{i\rightarrow\infty}
  \sum_{|u|=i} \frac{1}{L^i}\fa(pref_j(v)uw)\\
 & =  \lim_{j\rightarrow\infty}\lim_{i\rightarrow\infty}
  \sum_{|u|=i} \frac{1}{L^i}g(pref_j(v)uw) \\
 & =  g(v),
\end{align*}
where the last equality follows from the uniform continuity of $g$.
\end{proof}
\begin{corollary}
Every continuous function $\So\to{\mathbb R}$ that can be computed by a weighted
finite automaton can be computed by some average preserving weighted finite
automaton.
\end{corollary}

\subsection{Continuity of $\omega$-functions}
We now prove several results about WFA with continuous $\omega$-functions.
While the behavior of general WFA with continuous $\omega$-functions can be
complicated, we can obtain useful results for uniformly convergent WFA and the
uniform convergence assumption is well justified:
As we show in Lemma~\ref{sup}, all ap WFA with continuous
$\omega$-function are uniformly convergent. Together with Theorem~\ref{redistribution} 
we then have that uniformly convergent WFA compute all
WFA-computable continuous functions.

\begin{lemma}\label{sup}
Let $F$ be an ap word function. Let its $\omega$-function $f$ be continuous.
Then $F$ is uniformly convergent.
\end{lemma}
\begin{proof}
Let $\varepsilon >0$. By continuity of $f$ and compactness of $\So$, there
exists an index $k$
such that $|f(w)-f(w')|<\varepsilon$ for every $w,w'$ for which $pref_k(w')=pref_k(w)$.
Fix any $w\in\So$ and let $v\in\Ss$ be its prefix whose length is at least $k$.

By the ap property of $F$, we obtain:
\[
\frac{1}{|\Sigma|}\sum_{a\in\Sigma}F(va)=F(v).
\]
This means that
\[
\max\{F(va)\,|\,a\in\Sigma\}\geq F(v) \geq \min\{F(va)\,|\,a\in\Sigma\}.
\]

So for some letters $a,b\in\Sigma$, we have $F(va)\geq F(v)\geq F(vb)$. We can now
continue in this manner, obtaining words $w_1,w_2\in\So$ such that
$f(vw_1)\geq F(v)\geq f(vw_2)$.
By the choice of $k$, we have $f(w)+\varepsilon> f(vw_1)$ and $f(vw_2)>
f(w)-\varepsilon$. Therefore, $|f(w)-F(v)|< \epsilon$ and the claim follows.
\end{proof}

\begin{remark}
While most of the theorems in this section deal with uniformly convergent
functions and automata, uniform convergence is difficult to verify.
The ap property of minimal automata, on the other hand, is easy to check.
Thanks to Lemma~\ref{sup}, we can rewrite
all following theorems by replacing the uniform convergence assumption on $F$ by the demand that
$F$ be ap and $f$ be continuous. This is how we will mostly use
the results of this section, as ap WFA are often used in applications. 
However, it turns out that uniform convergence is
the essential feature that makes the following theorems valid, so we present the proofs in this more
general setup.
\end{remark}

We now state a simple but important property of uniformly convergent word functions.

\begin{lemma}\label{apu} Let $F$ be a word function defining the
$\omega$-function $f$. If $F$ is uniformly convergent, then
the following equality holds for all $w\in\So$ and $u\in\Ss$:
\[
\lim_{k\rightarrow\infty}F(pref_k(w)u)=f(w).
\]
\end{lemma}
\begin{proof}
By Lemma~\ref{lemStrongerUC}, for any $\epsilon>0$ there exists $k_0$ such that for all $k>k_0$
we have
$|F(pref_k(w)u)-f(w)|<\epsilon$.
\end{proof}

\begin{example}
Example~\ref{universal-counterexample}
shows that average preservation is a necessary condition in Lemmas~\ref{sup} and \ref{apu}.
The WFA $\A$ in Figure~\ref{u-c-figure}
is minimal but not average preserving. It defines the (continuous) zero $\omega$-function, but
the convergence is not uniform. Likewise, the conclusion of Lemma~\ref{apu}
also does not hold for $\F$:
Choose
$w=0^{\omega}$ and $u=1$. Then $f(0^\omega)=0$ while $\ds\lim_{k\to\infty}F(0^k1)=1$.
\end{example}

Next, we show that changing the initial distribution of a left minimal WFA does not
alter convergence and continuity properties:

\begin{lemma}\label{initial} Let $\A$ be a left minimal WFA and let $\B$ be a
WFA obtained from $\A$ by changing the initial distribution. Then the following
holds:
\begin{enumerate}[(1)]
\item If $\fa$ is defined on the whole $\So$ then so is $\fb$.
\item If $\fa$ is continuous then so is $\fb$.
\item If $\A$ is uniformly convergent then so is $\B$.
\end{enumerate}
\end{lemma}
\begin{proof}
Let us obtain $\B$ from $\A$ by changing the initial distribution to $I'$.
By the left minimality of $\A$, there are words
$u_i\in\Ss$ and coefficients $\alpha_i\in\mathbb R$ such that $I'=\alpha_1IA_{u_1}+\ldots +\alpha _nIA_{u_n}$. Then
$\B$ computes a function which is a linear combination of the functions
${\F}_{u_i}$:
\begin{align*}
\FB(v)&=\alpha_1 IA_{u_1}A_{v}F+\ldots
+\alpha _n IA_{u_n}A_{v}F\\
&=\alpha_1\F(u_1v)+\ldots+\alpha_n\F(u_nv)
\end{align*}
So, assuming that $\fa$ is defined everywhere, we obtain:
\[
\fb(w)=\alpha_1\fa(u_1w)+\ldots+\alpha_n\fa(u_nw),
\]
proving (1). Moreover, it is easy to observe (2) and (3) from these equalities.
\end{proof}
It follows from Lemma~\ref{initial} that if $\fa $ is a left minimal continuous
WFA then the sequence $\{A_{pref_k(w)}F\}_{k=1}^\infty$ of vectors
converges element-wise as $k$ tends to infinity:  To see that
$\{(A_{pref_k(w)}F)_i\}_{k=1}^\infty$ converges, we change the
initial distribution to the $i$-th element of the canonical basis
$I_i=(0,\ldots,0,1,0,\ldots,0)$.
Denote the resulting automaton by $\A_i$.  Then $\fai(w)$ is continuous
and
\[
\fai (w)=\lim _{k\rightarrow\infty}(I_iA _{pref_k(w)}F)=
\lim_{k\rightarrow\infty}(A _{pref_k(w)}F)_i.
\]
We see that $\ds\lim_{k\rightarrow\infty}A_{pref_k(w)}F$ is the vector with $i$-th
component equal to $\fai (w)$ for $i=1,\ldots,n$.

We now look into the effect of changing the final distribution of a right minimal
WFA. If the WFA is uniformly convergent then the outcome is the same as multiplying the
$\omega$-function by a constant.

\begin{lemma}\label{final} Let $\A$ be right minimal and uniformly
convergent. Then changing
the final distribution of $\A$ keeps uniform convergence and
affects $\fa$ by a multiplicative constant only.
\end{lemma}
\begin{proof}
Let $F'$ be any final distribution, and let $u_1,\ldots ,u_n\in\Ss$ be words
such that $F'=\alpha _1A_{u_1}F+\ldots +\alpha _nA_{u_n}F$ for some $\alpha _1,\ldots ,\alpha _n$.
Such words exist by the right minimality of $\A$. Denote by $\B$ the WFA $\A$ with
the final distribution $F'$. Then
\begin{align*}
\fb(w)  & =  \ds\lim_{k\rightarrow\infty}(IA_{pref_k(w)}F')\\
        & =  \ds\lim_{k\rightarrow\infty}(IA_{pref_k(w)}(\alpha _1A_{u_1}F+\ldots +\alpha _nA_{u_n}F))\\
        & =  \alpha _1\ds\lim_{k\rightarrow\infty}(IA_{pref_k(w)}A_{u_1}F)+\ldots +
\alpha _n\ds\lim_{k\rightarrow\infty}(IA_{pref_k(w)}A_{u_n}F)\\
        & =  (\alpha _1+\ldots +\alpha _n)\fa(w)
\end{align*}
where we have used Lemma~\ref{apu} in the last equality.

Uniform convergence of $\B$ easily follows, as the functions $F_i(v)=\F(vu_i)$
are all uniformly convergent and $\FB=F_1+F_2+\dots+F_n$.
\end{proof}
Putting Lemmas~\ref{initial} and \ref{final} together, we obtain a theorem about
the continuity of $\omega$-functions.
\begin{theorem}\label{keep-cont}
Let $\A$ be a minimal uniformly convergent WFA.
Then any automaton $\B$ obtained from $\A$ by changing $I$ and $F$ is also
uniformly convergent (and therefore continuous).
\end{theorem}
\begin{proof}
To prove the theorem, we change first $I$ and then $F$.

Lemma~\ref{initial} tells us that changing $I$ does not break uniform
convergence of $\A$. Also, it is easy to observe that changing $I$ does
not affect right-minimality of $\A$, so the conditions of Lemma~\ref{final} are
satisfied even after a change of initial distribution. Recall that uniform
convergence implies continuity by Lemma~\ref{remUCC}.
\end{proof}

Recall that for $w\in\So$ we define $A_w=\lim _{k\rightarrow\infty}
A_{pref_k(w)}$ if the limit exists. We are now prepared to prove that the weight
matrices of a minimal uniformly convergent WFA form a continuous RCP set.

\begin{corollary}\label{Awlimit}
Let $\A$ be a minimal uniformly convergent WFA. Then the
limit \[ A_w=\lim _{k\rightarrow\infty} A _{pref_k(w)}\] exists for all
$w\in\So$, the elements of the matrix $A_w$ are continuous functions of $w$ and we have $\fa(w)=IA_wF$.
\end{corollary}
\begin{proof}
Taking $I_i=(0,\dots,0,1,0,\dots,0)$ and $F_j=(0,\dots,0,1,0,\dots,0)^T$ with
one on the $i$-th and $j$-th place, we obtain the automaton $\A_{ij}$ computing
$(A_w)_{ij}$. From Lemma~\ref{keep-cont}, we
see that ${\fa}_{ij}$ is continuous on the whole $\Sigma$, so elements of
$A_w$ are continuous functions of $\Sigma$.

Multiplications by constant vectors $I$ and $F$ are continuous operations so we
can write
\[
\fa(w)= \lim _{k\rightarrow\infty} IA _{pref_k(w)}F = I\left(\lim
_{k\rightarrow\infty} A _{pref_k(w)}\right)F
=IA_wF,
\]
concluding the proof.
\end{proof}

Next we look into the matrices $A_w$ and prove that they have some very particular properties.

\begin{lemma}\label{wjauonw}
Let $\A$ be minimal and uniformly convergent, $w\in\So$ and $u\in\Ss$. Then
\[
A_wA_u=A_w.
\]
\end{lemma}
\begin{proof}
By Theorem~\ref{keep-cont} we can change the initial and final distributions of $\A$ to any
$I$ and $F$ without
affecting uniform convergence. Then Lemma~\ref{apu} gives us that
\[
IA_wA_uF=
I (\lim_{k\rightarrow\infty}A_{pref_k(w)}) A_uF=
\lim_{k\rightarrow\infty}(IA_{pref_k(w)}A_uF)=\lim_{k\rightarrow\infty}(IA_{pref_k(w)}F)=IA_wF.\]
As the above equality holds for all $I$ and $F$, we have
\[A_wA_u=A_w.\qedhere
\]
\end{proof}

\begin{corollary} \label{cons}Let $\A$ be minimal and uniformly convergent. If $\fa$ is a non-zero
function, then we can effectively find a vector $I_c\neq 0$ such
$I_cA_a=I_c$ for all $a\in\Sigma$ . 
\end{corollary}
\begin{proof}
Suppose $\fa\neq 0$. Then $I A_w F\neq 0$ for some $w\in\So$.
Let $I_c=IA_w$. Consider now the WFA $\B$ obtained from $\A$ by replacing
the initial distribution $I$ with $I_c$. Then,
by Lemma~\ref{wjauonw}, we have for all $u\in\Ss$:
\[
I_cA_uF=IA_wA_uF=IA_wF=I_cF\neq 0.
\]
Thus $\B$ computes a non-zero constant function.

Next we notice that for all $u\in\Ss$ and all $a\in\Sigma$ we have the equality
$I_cA_aA_uF=I_cF=I_cA_uF$. This together with the right minimality of $\A$
gives us that $I_cA_a=I_c$ for all $a\in\Sigma$.

We have shown that the matrices $A_a,\,a\in\Sigma$ always have a common
left eigenvector belonging to the eigenvalue 1. 
We can find such common left eigenvector $I_c$
effectively by solving the set of linear
equations $\{I_c(A_a-E)=0,\,a\in\Sigma\}$.
\end{proof}

\begin{remark}
It is easy to see from  
Corollary~\ref{cons} that any minimal and uniformly convergent WFA that
computes a non-zero function  can be made to compute a nonzero
\emph{constant} function just by changing its initial distribution to $I_c$. 
\end{remark}

If $\A$ is uniformly convergent minimal WFA, then the rows of all
limit matrices $A_w$ are multiples of the same vector.
 
\begin{lemma}\label{onedim} Let $\A$ be minimal uniformly convergent
WFA and let $\fa\neq 0$. Then for all $w\in\So$
the row space $V(A_w)$ of $A_w$ is one-dimensional. Moreover,
$V(A_w)$ is the same for all $w\in\So$.
\end{lemma}
\begin{proof}
By Lemma~\ref{wjauonw}, $A_wA_uF=A_wF$ and thus
$A_w(A_uF-F)=0\textrm{ for all }u\in\Ss.$ We see that vector $A_uF-F$ is
orthogonal to $V(A_w)$ irrespective of the choice of $u$.

Denote $W=\langle
A_uF-F|F\in \Ss \rangle$. Now the minimality of $\A$ implies $\dim W \geq n-1$,
because $\dim (W+\langle F\rangle) =n$. On the other hand, for all $w\in\So$
we have $V(A_w)\subseteq W^{\bot}$
so $\dim V(A_w)\leq \dim W^{\bot} \leq 1$. If $A_w=0$, for some $w$
then $\fa(uw)=IA_uA_wF=0$ for all $u\in\Ss$ and so, by continuity of $\fa$,
we would have $\fa=0$. This means that $\dim V(A_w)=1$ and
$V(A_w)=W^{\bot}$ for all $w$.
\end{proof}

\begin{remark}
From Lemma~\ref{onedim} it follows that the vector $I_c$ from Corollary~\ref{cons} 
belongs to $V(A_w)$ and is therefore unique up to multiplication by a
scalar.
\end{remark}

\begin{remark} Let $\A$ be minimal and uniformly convergent.
If $F$ is an eigenvector belonging to the eigenvalue $\lambda$ of some $A_u$,
then for all $w\in\So$
\[A_wF=A_wA_uF=A_w\lambda F=\lambda A_wF,\]
and thus either $\lambda =1$ or $A_wF=0$ for all $w\in\So$. In the latter case
$\fa = 0$.
\end{remark}

Using Corollary~\ref{cons} and Lemma~\ref{onedim}, we can transform all
minimal uniformly convergent automata to a ``canonical form''. This
transformation is a simple change of basis, so it preserves minimality as well
as the the ap property:

\begin{lemma}\label{matrix1} Let $\A$ be a minimal uniformly convergent
automaton such that
$\fa\neq 0$. Then we can algorithmically find a basis of
$\mathbb{R}^{n}$ such that the matrices $A_a$ are all of the form
\begin{equation}\label{preform}
A_a=
\begin{pmatrix}
 B_a & \vline & \mathbf{b}_a \\ \hline
  \mathbf{0} & \vline & 1
\end{pmatrix},
\end{equation}
where $\{B_a\,|\,a\in\Sigma\}$ is a stable set of matrices.
\end{lemma}
\begin{proof}
Suppose that $\A$ is minimal and uniformly convergent. Using 
Corollary~\ref{cons}, we can algorithmically find a vector 
$I_c$ such that $I_cA_a=I_c$ for all $a\in \Sigma$.
Let
us change the basis of the original automaton so that $I_c=(0,\dots,0,1)$ (this does not affect 
uniform convergence or minimality of $\A$).

As we have $I_cA_a=I_c$, the lowest row of every weight matrix $A_a$ must
be equal to $(0,\ldots,0,1)$. In other words, we have shown that for every $a\in\Sigma$, matrix $A_a$ has the form
\[
A_a=
\begin{pmatrix}
 B_a & \vline & \mathbf{b}_a \\ \hline
  \mathbf{0} & \vline & 1
\end{pmatrix},
\]
where $\mathbf{0}$ and $\mathbf{b}_a$ are row and column vectors, respectively. For
$v$ word (finite or infinite), denote
\[
A_v=
\begin{pmatrix}
 B_v & \vline & \mathbf{b}_v \\ \hline
  \mathbf{0} & \vline & 1
\end{pmatrix}.
\]

From the formula for matrix multiplication, we obtain
\[
A_{uv}=A_u\cdot A_v=
\begin{pmatrix}
 B_u & \vline & \mathbf{b}_u \\ \hline
  \mathbf{0} & \vline & 1
\end{pmatrix}\cdot
\begin{pmatrix}
 B_v & \vline & \mathbf{b}_v \\ \hline
  \mathbf{0} & \vline & 1
\end{pmatrix}=
\begin{pmatrix}
 B_u B_v & \vline & B_u\mathbf{b}_v+\mathbf{b}_u \\ \hline
  \mathbf{0} & \vline & 1
\end{pmatrix},
\]
in particular $B_{uv}=B_{u}B_{v}$ and so $B_v$ is simply a product
$B_{v_1}B_{v_2}\cdots B_{v_n}$.

For all $w\in\So,$ we have:
\[A_w=
\begin{pmatrix}
 B_w & \vline & \mathbf{b}_w \\ \hline
  \mathbf{0} & \vline & 1
\end{pmatrix}.\]

By Lemma~\ref{onedim}, we know that if $\A$ defines a continuous
function, then the rows $1,\ldots,n-1$ in $A_w$ are multiples of
row $n$. This means that $B_w=\mathbf{0}$ for all $w\in\So$ and so
$\{B_a\,|\,a\in\Sigma\}$ is a stable set.
\end{proof}

One might ask if all automata with matrices of the form (\ref{preform}) and
$\{B_a\,|\,a\in\Sigma\}$ stable are uniformly convergent. We show that the
answer is yes and prove an even more general statement along the way
(we are going to need this more general form later when proving Theorem~\ref{equiv}).

\begin{lemma}\label{matrix2}
Let $\{A_a\,|\,a\in\Sigma\}$ be a finite set of matrices of the form
\[
A_a=
\begin{pmatrix}
 B_a & \vline & C_a \\ \hline
 \bf{0} & \vline & D_a
\end{pmatrix},
\]
where $B_a$ and $D_a$ are square matrices and the set $\{B_a\,|\,a\in\Sigma\}$ is
stable.
Then the following holds:
\begin{enumerate}[(1)]
\item If $\{D_a\,|\,a\in\Sigma\}$ is product-bounded then $\{A_a\,|\,a\in\Sigma\}$ is product-bounded.
\item If $\{D_a\,|\,a\in\Sigma\}$ is RCP then $\{A_a\,|\,a\in\Sigma\}$ is RCP.
\item If $\{D_a\,|\,a\in\Sigma\}$ is continuous RCP then $\{A_a\,|\,a\in\Sigma\}$ is continuous RCP.
\item If $\{D_a\,|\,a\in\Sigma\}$ is stable then $\{A_a\,|\,a\in\Sigma\}$ is stable.
\end{enumerate}
\end{lemma}
\begin{proof}
As before, denote
\[
A_v=
\begin{pmatrix}
 B_v & \vline & C_v \\ \hline
 \bf{0} & \vline & D_v
\end{pmatrix}.
\]
It is easy to see that $B_v$ and $D_v$ are equal to the matrix products $B_{v_1}B_{v_2}\dots B_{v_n}$
and $D_{v_1}D_{v_2}\dots D_{v_n}$, respectively, while for $C_v$ the
equality $C_{uv}=B_uC_v+C_uD_v$ holds.

\begin{enumerate}[(1)]
\item
Let $K$ be a constant such that $\|D_u\|<K$ for all $u\in\Ss$. We need to prove
that there exists a constant $L$ such that $\|C_u\|<L$ for all $u\in\Ss$.

By Lemma~\ref{uniform}, there exists a $k$ such that for all words $u$ of
length at least $k$ we have $\|B_u\|<1/2$. Denote $M=\max\{\|
B_u\|\,|\,u\in\Ss, |u|< k\}$ and $N=\max \{\|C_a\|\,|\, a\in\Sigma\}$.

Let $m=|u|$. It is easy to see that when $m\geq k$, the inequality
$\|B_u\|<M\cdot 2^{-\lfloor
m/k\rfloor}$ holds. Moreover, a quick proof by induction yields that:
\[
C_u=\sum_{j=0}^m B_{u_1\cdots u_{j-1}}C_{u_j}D_{u_{j+1}\cdots u_m}.
\]
Hence, we can write (for $m>k$):
\begin{align*}
\|C_u\|&\leq \sum_{j=0}^{k-1} \|B_{u_1\cdots u_{j-1}}\|\|C_{u_j}\|\|D_{u_{j+1}\cdots
u_m}\|\,+\\
&\quad\quad\quad+\sum_{j=k}^{m} \|B_{u_1\cdots u_{j-1}}\|\|C_{u_j}\|\|D_{u_{j+1}\cdots
u_m}\|\\
&\leq \sum_{j=0}^{k-1} M N K +\sum_{j=k}^{m} M\cdot 2^{-\lfloor
j/k\rfloor}\cdot N K
\end{align*}
The first sum is exactly $kMNK$ while the second one can bounded from
the above by $kMNK$. All in all, we obtain that $\|C_u\|\leq 2kMNK$ and so
$\{A_a\,|\,a\in \Sigma\}$ is product-bounded.

\item
Using Theorem~\ref{wang2}, we obtain that the set
$\{D_a\,|\,a\in\Sigma\}$ is product-bounded. Therefore, using the part
(1) of this Lemma, we see that $\{A_a\,|\,a\in\Sigma\}$ is
product-bounded and so there exists some $L>0$ such that $\|C_u\|<L$
for all $u\in\Ss$.

We only need to show that for every $w\in\So$, the sequence
$\{C_{pref_k(w)}\}_{k=1}^\infty$ satisfies the
Bolzano-Cauchy condition.

Assume $\epsilon>0$ is given. Because $\{B_a\,|\,a\in\Sigma\}$ is stable, there exists a $k$
such that $\|B_{pref_k(w)}\|<\epsilon/(4L)$. Denote $u=pref_k(w)$ and let
$x\in\So$ be such that $w=ux$. The sequence $\{D_{pref_i(x)}\}_{i=1}^\infty$ converges so there exists
a number $j$ such that for all positive $i$ we have $\|D_{pref_j(x)} - D_{pref_{j+i}(x)}\| <
\epsilon/(2L)$.
Let $v=pref_j(x)$ and write $w=uvy$ where $y$ is an appropriate infinite suffix.

We will now prove that for all prefixes $uvs$ of $w$ we have $\|C_{uv} -
C_{uvs}\| < \epsilon$.
Using the equalities
\begin{align*}
     C_{uv}  &= B_uC_v    + C_uD_v\\
     C_{uvs} &= B_uC_{vs} + C_uD_{vs},
\end{align*}
we obtain
$$ \|C_{uv} - C_{uvs}\| \leq \|B_u\| \|C_v-C_{vs}\| + \|C_u\| \|D_v-D_{vs}\| <
\frac{\epsilon}{4L}\cdot 2L + L\cdot\frac{\epsilon}{2L} = \epsilon,$$
this means that the sequence $\{C_{pref_k(w)}\}_{k=1}^\infty$ is
Cauchy and so the proof is finished.
\item
By case (2) we have that $A_w$ exists for all $w\in\So$. As $B_w$, $D_w$ depend
continuously on $w$, all we need to show is that the map $w\mapsto C_w$ is also continuous.

As before, by Theorem~\ref{wang2} the set $\{C_a\,|\,a\in\Sigma\}$ is
product-bounded. By passing to limits, we see that there exists $L$
such that $\|C_w\|<L$ for all infinite $w\in\So$.

The function $w\mapsto D_w$ is continuous on a compact
space and so it is uniformly continuous. Given $\epsilon>0$, we find $k$ such
that for all $u,v$ of length $k$ and all $w,z\in\So$ we have:
\begin{align*}
\|B_u\|<\frac{\epsilon}{4L}\\
\|D_{vw}-D_{vz}\|<\frac{\epsilon}{2L}.
\end{align*}

We can now, similarly to case (2), write:
\begin{align*}
\|C_{uvz} - C_{uvw}\|& \leq \|B_u\| \|C_{vz}-C_{vw}\| + \|C_u\|
\|D_{vz}-D_{vw}\|\\
&<\frac{\epsilon}{4L}\cdot 2L + L\cdot\frac{\epsilon}{2L} = \epsilon,
\end{align*}
proving continuity.

\item
Using case (2), we obtain that $C_z$ exists for all $z\in\So$ and
moreover, by Theorem~\ref{wang2}, there exists $L>0$ such that
$\|C_z\|<L$ for all $z\in\So$.

Let $w\in\So$ and $\epsilon>0$. If we prove that $\|C_w\|<\epsilon$, we are
done. There is a finite prefix $u$ of $w$
such that $\|B_u\|<\epsilon/L$. Let $w=uz$, where word $z\in\So$ is the remaining infinite
suffix of $w$. We now have:
\[
\|C_w\|=\|C_{uz}\|=\|B_uC_z+C_uD_z\|=\|B_uC_z\|\leq \|B_u\|\|C_z\|<\frac{\epsilon}{L}L=\epsilon,
\]
where we have used the equality $D_z=0$. This means that $\|C_w\|=0$ and we are
done.\qedhere
\end{enumerate}
\end{proof}

Observe that by letting $D_a=1$ for all $a\in\Sigma$, we
obtain from case (3) of Lemma~\ref{matrix2} and Corollary~\ref{cont-conv} a
partial converse to Lemma~\ref{matrix1}: All automata of the form (\ref{preform}) with
$\{B_a\,|\,a\in\Sigma\}$ stable are uniformly convergent.

Therefore, putting Lemmas \ref{matrix1} and \ref{matrix2} together, we obtain
(under the assumption that $\A$ is minimal and $\fa\neq0$) a characterization
of uniformly convergent automata.

\begin{corollary}\label{matrix-reloaded}
Let $\A$ be a minimal automaton such that $\fa\neq0$. Then $\A$ is uniformly
convergent iff there exists a basis of ${\mathbb R}^n$ in which
all the transition matrices $A_a$
have the form (\ref{preform}) where $\{B_a\,|\,a\in\Sigma\}$ is a stable set.
\end{corollary}

Note that we could have relied on Theorem~\ref{Daubechies} here: Together
with Lemma~\ref{WFA-RCP}, it directly gives us Corollary~\ref{matrix-reloaded}.
(We only need to realize that the dimension of the space $E_1$ is one in this
case, which follows from Lemma~\ref{onedim}.) However, we wanted to show how to
algorithmically obtain the form $(\ref{preform})$ and we will also need Lemma
\ref{matrix2} later on.

\begin{remark}
If $\A$ is a minimal ap automaton then it is easy to verify algorithmically
whether $\fa=0$, because $\fa=0$ iff $\F=0$.
\end{remark}

\begin{remark}
In \cite{Karh}, the authors define \emph{level automata} as automata
satisfying the following conditions:
\begin{enumerate}[(1)]
\item Only loops of length one (i.e. $q\to q$) are allowed.
\item The transition matrices and distribution vectors are non-negative.
\item For every state $p$, if there exist a state $q\neq p$ and a letter $a$ such
that $(A_a)_{p,q}\neq 0$ then $(A_a)_{p,q}<1$ for all $q$ and $a$. If there are no
such $q$ and $a$ then $(A_a)_{p,p}=1$ for all letters $a$.
\item The automaton is reduced; it does not have useless states.
\end{enumerate}

As all level automata have (after proper ordering of states) matrices of the
form 
\[
A_a=
\begin{pmatrix}
 B_a & \vline & C_a \\ \hline
 \bf{0} & \vline & E
\end{pmatrix},
\]
where $E$ is the unit matrix and $B_a$ are upper triangular matrices with
entries in the interval $[0,1)$, the automata from case (3) of Lemma
\ref{matrix2} are actually a generalization of level automata.
\end{remark}

\subsection{WFA and RCP sets}

In this part we explicitly connect the notions of RCP sets and functions
computed by WFA.
\begin{theorem}\label{WFA-RCP}
Let $\A$ be a WFA and let $B=\{A_a\,|\,a\in\Sigma\}$ be its set of transition matrices.
Then the following holds:
\begin{enumerate}[(1)]
\item If $B$ is an RCP set then $\fa$ is defined everywhere.
\item If $B$ is a continuous RCP set then $\A$ is uniformly convergent (and
therefore $\fa$ is continuous).
\end{enumerate}
For the converse, we need to assume minimality:
\begin{enumerate}[(1)]
\item[3.] If $\A$ is uniformly convergent and minimal then $B$ is a continuous RCP set.
\end{enumerate}
\end{theorem}
\begin{proof}
\begin{enumerate}[(1)]
\item If the limit $A_w=\lim_{k\to\infty} A_{pref_k(w)}$ exists then
\[\fa(w)=\lim_{k\to\infty} \F(pref_k(w))=I A_w F.\]
As
$A_w$ is defined everywhere, so is $\fa$.
\item Similarly to the first proof, we have $\fa(w)=IA_wF$ where
$w\mapsto A_w$ is a continuous function, so $w\mapsto \fa(w)$ is continuous. Uniform convergence
follows from Corollary~\ref{cont-conv}, continuity from Lemma~\ref{remUCC}.
\item This is precisely Corollary~\ref{Awlimit}.\qedhere
\end{enumerate}
\end{proof}
The uniform convergence and minimality conditions in the third statement are both necessary,
as we can see from the following two examples where $\fa$ is continuous but $A$ is
not even RCP:
\begin{example}
We construct a counterexample that is ap (and thus uniformly convergent by
Lemma~\ref{sup}) but not minimal. Let
$I=(0,1),F=(0,1)^T$ and
\[
A_0=A_1=\begin{pmatrix}
-1&0\\
0&1\\
\end{pmatrix}
\]
This automaton is ap and computes the constant function $\fa(w)=1$, yet the set
$\{A_0,A_1\}$ is not RCP.
\end{example}
\begin{example}
To obtain a minimal automaton that computes a continuous function, but does not
have RCP set of transition matrices, take the
automaton $\A$ from Example~\ref{universal-counterexample}. This automation computes
the zero $\omega$-function and has transition matrices
\[
A_0=\begin{pmatrix}
-1&0\\
0&0\\
\end{pmatrix}
\qquad
A_1=\begin{pmatrix}
0&1\\
0&0\\
\end{pmatrix}.
\]
Now observe that
\[
A_0^n=\begin{pmatrix}
(-1)^n&0\\
0&0\\
\end{pmatrix},\]
so $\A$ is not RCP.
\end{example}

The next example shows that even if $\A$ is minimal and  ap, and if $\fa$ is everywhere defined
and continuous everywhere except at one point, we can
not infer that $A$ is RCP.

\begin{example}
Let $I=(1,0),F=(0,1)^T$ and
\[
A_0=\begin{pmatrix}
-1&0\\
0&1\\
\end{pmatrix},\quad
A_1=\begin{pmatrix}
0&1\\
0&1\\
\end{pmatrix},\quad
A_2=\begin{pmatrix}
0&-1\\
0&1\\
\end{pmatrix}.
\]
It is easy to see that $\A$ is both ap and minimal. Moreover, we have
\begin{align*}
F(0^n)&=(-1)^n IF=0\\
F(0^n1w)&=(-1)^n\\
F(0^n2w)&=(-1)^{n+1}
\end{align*}
for every $w\in\Ss$. This means that $\fa$ is defined everywhere. However,
$\fa$ is not continuous at $0^\omega$.
The set $A$ is not RCP, because we have
\[
A_0^n=\begin{pmatrix}
(-1)^n&0\\
0&1\\
\end{pmatrix}.
\]
\end{example}

\subsection{Decision problems for $\omega$-functions}
\label{omegadecidability}
In this part, we present several results about decidability of various
properties of the $\omega$-function $\fa$ in the case of ap automata.
In particular, we are interested to know how to determine if the
$\omega$-function $\fa$ is everywhere defined, or everywhere continuous.
It turns out that the questions are closely related to the decidability
status of the matrix stability problem: If it is undecidable whether a given
finite set of matrices is stable then it is also undecidable for a given ap WFA $\A$
whether $\fa$ is everywhere defined, or whether $\fa$ is continuous. We also show that
in this case it is undecidable if a given finite matrix set is RCP, or if it is continuous
RCP. Conversely, if it were the case that stability is decidable then continuity of
$\fa$ is decidable, as is the question of whether a given matrix set is continuous RCP.
The central algorithmic problem is therefore the following:

\bigskip
\noindent
{\sc Matrix Product Stability}:

{\bf Input:} A finite set $\{A_a\,|\,a\in\Sigma\}$ of $n\times n$ matrices.

{\bf Question:} Is $\{A_a\,|\,a\in\Sigma\}$ stable?
\bigskip

We begin with the equivalence problem of two ap WFA.
\begin{theorem}
\label{ap-equivalence}
Given two ap WFA $\A$ and $\B$ such that at least one of the $\omega$-functions $\fa$ and $\fb$
is everywhere defined, one can algorithmically decide whether $\fa=\fb$.
\end{theorem}
\begin{proof}
To decide $\fa=\fb$, we construct ap automaton $\C$ computing the difference $\fa-\fb$ and then
minimize $\C$, obtaining some automaton $\D$. Minimization is effective by Lemma~\ref{find-minimal}.
Now from Lemma~\ref{zero-ap-lemma}, we get that $\fa-\fb=0$ iff $\D$ is the trivial
automaton. Note that $\fa-\fb$ is not defined on those $w\in\So$ for which
exactly one of the functions $\fa$ and $\fb$ is undefined. Hence $\fa-\fb=0$
is equivalent to $\fa=\fb$.
\end{proof}
Note that the process in the previous proof fails if $\fa=\fb$ is not everywhere defined:
in this case also $\fa-\fb$ will be undefined for some $w\in\So$, yielding
(wrongly) a negative answer.

In contrast to Theorem~\ref{ap-equivalence}, if {\sc Matrix Product Stability}
is undecidable then the analogous question is undecidable without the ap assumption.
In this case one cannot even determine if a given non-ap WFA defines the zero-function.

\begin{theorem}
\label{is-zero}
{\sc Matrix Product Stability} is algorithmically reducible to the problem
of determining if $\fa=0$ for a given WFA $\A$.
\end{theorem}
\begin{proof}
Given a set of matrices $B=\{A_a\,|\,a\in\Sigma\}$, we construct automata
$\A_{ij}$ with transition matrices $A_a$, initial distribution
$I_i$ and final distribution $I_j^T$ (where $I_1,\dots,I_n$ is a basis of
$\mathbb{R}^n$). Obviously, $B$ is stable iff all the
$\omega$-functions computed by $\A_{ij}$ are zero.
\end{proof}

We conjecture that Theorem~\ref{is-zero} holds even under the additional
assumption that $\fa$ is known to be everywhere defined and continuous, but we can
not offer a proof.

Recall that Theorem~\ref{redistribution} tells us that for every WFA computing
a continuous function there is an ap WFA that computes the same function. It
would be interesting to know whether this conversion can be done effectively.
One consequence of Theorems~\ref{ap-equivalence} and \ref{is-zero} is that, assuming {\sc Matrix Product
Stability} is undecidable, we cannot effectively convert a non-ap WFA into an
ap WFA with the same $\omega$-function.

In the following we reduce {\sc Matrix Product Stability} to the following decision problems:
\bigskip

\noindent
{\sc Ap-WFA convergence}:

\nobreak
{\bf Input:} An average preserving WFA $\A$.

\nobreak
{\bf Question:} Is $\fa$ everywhere defined?
\bigskip

\noindent
{\sc Ap-WFA continuity}:

\nobreak
{\bf Input:} An average preserving WFA $\A$.

\nobreak
{\bf Question:} Is $\fa$ everywhere continuous?
\bigskip

\noindent
{\sc Matrix Product Convergence}:

\nobreak
{\bf Input:} A finite set $\{A_a\,|\,a\in\Sigma\}$ of $n\times n$ matrices.

\nobreak
{\bf Question:} Is $\{A_a\,|\,a\in\Sigma\}$ an RCP set?
\bigskip

\noindent
{\sc Matrix Product Continuity}:

\nobreak
{\bf Input:} A finite set $\{A_a\,|\,a\in\Sigma\}$ of $n\times n$ matrices.

\nobreak
{\bf Question:} Is $\{A_a\,|\,a\in\Sigma\}$ a continuous RCP set?
\bigskip

To simplify our constructions, we use the fact that the problems {\sc Matrix Product Stability}, {\sc
Matrix Product Convergence} and {\sc Matrix Product Continuity} are as hard for
a pair of matrices as they are for any finite number of matrices, see
\cite{unsolvedblondel}. The elementary proof for {\sc Matrix Product Stability} we present is based
on~\cite{Blondel2}.

\begin{lemma}
\label{reductiontosizetwo}
The {\sc Matrix Product Stability} problem for a set
$\{A_1,A_2,\dots,A_m\}$ of matrices, is algorithmically reducible to {\sc
Matrix Product Stability} for a pair of matrices $\{B_0, B_1\}$.
\end{lemma}

\begin{proof}
For given $m$ matrices $A_1,A_2,\dots ,A_m$ of size $n\times n$ we construct
two matrices of size $mn \times mn$ that in the block form are
$$
\begin{array}{rclcrcl}
B_0 &=&
\left(
\begin{array}{c|c}
\mathbf{0} & E_{m(n-1)} \\
\hline
\mathbf{0} & \mathbf{0}
\end{array}\right)
& \hspace*{5mm} &
B_1 &=&
\left(
\begin{array}{c|ccc}
A_1 & \mathbf{0}&\dots&\mathbf{0}\\
\hline
\vdots&\vdots&&\vdots\\
\hline
A_{m} & \mathbf{0}&\dots&\mathbf{0}
\end{array}\right).
\end{array}
$$
Here $E_{m(n-1)} $ is the $m(n-1)\times m(n-1)$ identity matrix, and $\mathbf{0}$ indicates the
zero matrix of appropriate size.

In the same way that we produce graphs of WFA, we construct the graph
in Figure~\ref{matrixgraph}. (We are actually constructing a WFA over the ring
of $n\times n$ matrices.)

Consider now the matrix $B_v$ where $v\in\{0,1\}^*$. This matrix can be divided
into $m\times m$ blocks of size $n\times n$. To calculate the value of
the block at the position $i,j$, we add up all the products along
all paths labeled by $v$ from vertex $i$ to vertex $j$. Due to the shape of the
graph in Figure~\ref{matrixgraph}, there will be always at most one such
path for each $i,j,v$ and each $B_v$ will have at most $m$ nonzero $n\times
n$ blocks.
\picture{Directed graph whose paths correspond to blocks
in the products of $B_0$ and $B_1$ in the proof of
Lemma~\ref{reductiontosizetwo}.}{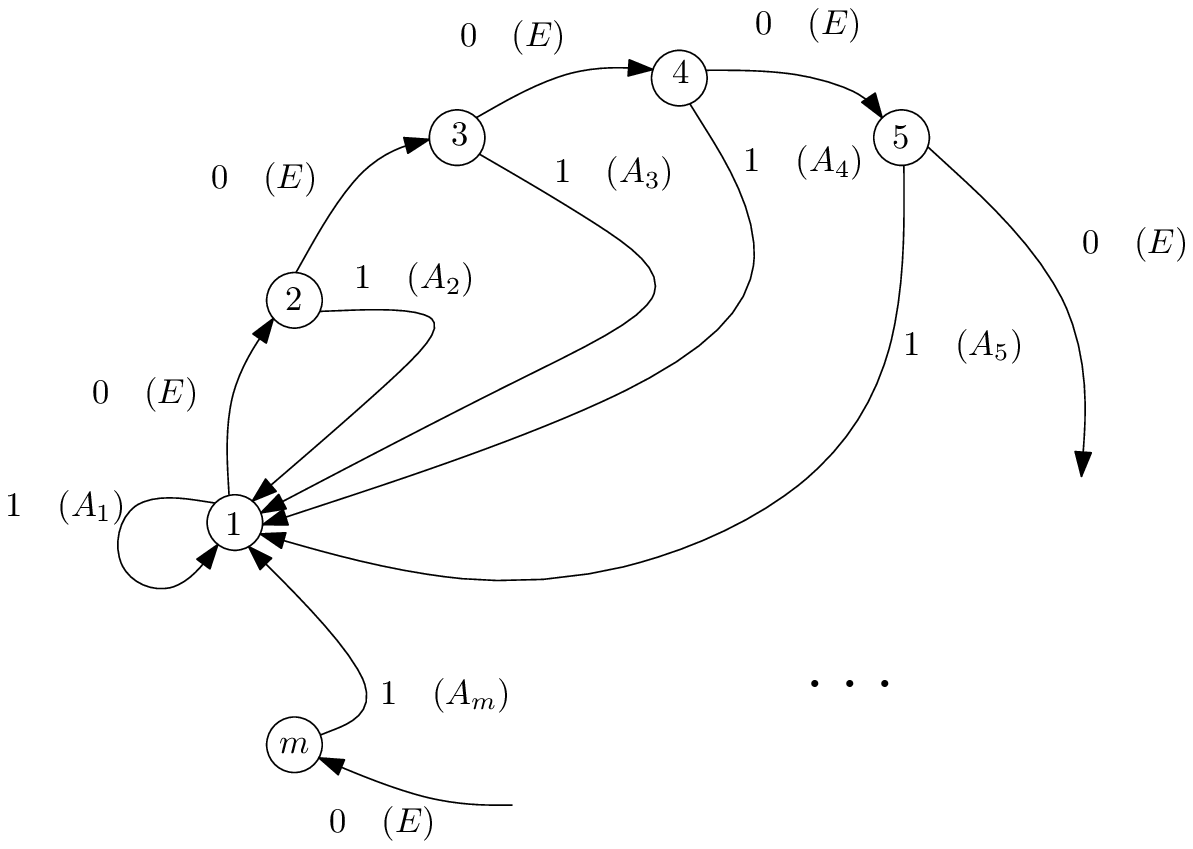}{matrixgraph}

Moreover, it is easy to see that the blocks in infinite products are exactly
all the infinite products of matrices $A_i$ (or zero matrices), so it is clear that $\{B_0, B_1\}$ is
stable if and only if $\{A_1,A_2,\dots,A_m\}$ is
stable.
\end{proof}

\begin{theorem}
\label{reduction1}
{\sc Matrix Product Stability} is algorithmically reducible to problems
{\sc Ap-WFA convergence},
{\sc Ap-WFA continuity},
{\sc Matrix Product Convergence} and
{\sc Matrix Product Continuity}
\end{theorem}
\begin{proof}
Let $B=\{B_a\,|\,a\in\Sigma\}$  be a set of matrices whose stability we want to
decide. Thanks to Lemma~\ref{reductiontosizetwo}, we can assume
$\Sigma=\{0,1\}$.

We create
several ap-automata $\A_{ij}$ such that:
\begin{itemize}
\item if $B$ is stable then  the function ${\fa}_{ij}$ is
continuous and the matrices of $\A_{ij}$ form a continuous RCP set for each $i,j$, while
\item if $B$ is not stable then for some $i,j$ the function ${\fa}_{ij}$ is
not everywhere defined and the transition matrices of $\A_{ij}$ are not an RCP set.
\end{itemize}
The result then follows directly.

We choose the transition matrices for $\A_{ij}$ as follows:
\[
A_0=\begin{pmatrix}
B_0&\vline& \mathbf{b}_{0,j}\\
\hline
\mathbf{0}&\vline& 1\\
\end{pmatrix},\quad
A_1=\begin{pmatrix}
B_1&\vline& \mathbf{b}_{1,j}\\
\hline
\mathbf{0}&\vline& 1\\
\end{pmatrix}.
\]
Here the column vectors $\mathbf{b}_{0,j}$ and $\mathbf{b}_{1,j}$ have all entries zero except for the $j$-th.
The $j$-th entry of $\mathbf{b}_{0,j}$ is $1$ while the $j$-th entry of
$\mathbf{b}_{1,j}$ is $-1$.

The initial distribution of $\A_{ij}$ is $I_i=(0,\dots,0,1,0,\dots,0)$ with
one on the $i$-th place. The final distribution is the same for all
automata: $F=(0,\dots,0,1)^T$.

First observe that
\[
(A_0+A_1) F=\begin{pmatrix}
\ds B_0+B_1&\vline& \mathbf{0}\\
\hline
\mathbf{0}&\vline& 2\\
\end{pmatrix}
\cdot
\begin{pmatrix}
0\\
\vdots\\
0\\
1\\
\end{pmatrix}
=\begin{pmatrix}
0\\
\vdots\\
0\\
2\\
\end{pmatrix}=|\Sigma|\cdot F,
\]
so each $\A_{ij}$ is ap. If $B$ is a stable set then from the case (3) of
Lemma~\ref{matrix2}
we obtain that for all $i,j$ the set $\{A_a\,|\,a\in\Sigma\}$ is a continuous RCP set, and therefore
all ${\fa}_{ij}$ are continuous.

Assume then that $B$ is not stable, so
for some $w$ the limit $\ds\lim_{k\to \infty} B_{pref_k(w)}$ is
not zero or does not exist. Then there exists a pair $(i,j)$ such that the
sequence 
$\left\{\left({B_{pref_k(w)}}\right)_{ij}\right\}_{k=1}^\infty$ does not converge to zero.
Consider the value of $\Faij(pref_k(w))$. The product of transition matrices will
be
\[A_{pref_k(w)}=
\begin{pmatrix}
 B_{pref_k(w)} & \vline & \mathbf{b}_k \\ \hline
  \mathbf{0} & \vline & 1
\end{pmatrix},
\]
where $\mathbf{b}_k$ is some column vector. The value of $\Faij(pref_k(w))$ is
equal to $I_i A_{pref_k(w)} F$, which, after a short calculation, turns out to be
the $i$-th element of $\mathbf{b}_k$.

Moreover, it is straightforward to verify that the vectors $\mathbf{b}_k$ satisfy the equation
$\mathbf{b}_{k+1}=\mathbf{b}_k+B_{pref_k(w)}\mathbf{b}_{w_{k+1},j}$. Taking
the $i$-th element of $\mathbf{b}_{k+1}$, we get the equation for
$\Faij(pref_k(w))$:
\[
\Faij(pref_{k+1}(w))=\Faij(pref_{k}(w))+c_{w_{k+1}}
\left({B_{pref_k(w)}}\right)_{ij},\]
where $c_{w_{k+1}}$ is the $j$-th element of $\mathbf{b}_{w_{k+1},j}$,
i.e. either 1 or $-1$. We obtain
\[
|\Faij(pref_{k+1}(w))-\Faij(pref_{k}(w))|=\left|\left({B_{pref_k(w)}}\right)_{ij}\right|.\]
Because the sequence
$\left\{\left({B_{pref_k(w)}}\right)_{ij}\right\}_{k=1}^\infty$ does not tend to zero, neither does
the difference $|\Faij(pref_{k+1}(w))-\Faij(pref_{k}(w))|$. But then
the sequence of values $\{{{\Fa}}_{ij}(pref_{k}(w))\}_{k=1}^\infty$ does not satisfy the Bolzano-Cauchy
condition and can not converge. Therefore, ${\fa}_{ij}(w)$ remains undefined.

By Theorem~\ref{WFA-RCP} the matrices of $\A_{ij}$ are then not an RCP set, which concludes the proof.
\end{proof}

\begin{remark}
\label{inseparable}
Note that the  proof showed, in fact, more: if {\sc Matrix Product Stability}
is undecidable then the continuous ap WFA  are recursively inseparable from
the ap WFA that are not everywhere defined. Recall that two disjoint sets $A,B$
are called recursively inseparable if there does not exist an algorithm that on input $x$
returns value 0 if $x\in A$, value 1 if $x\in B$ and may return either value if
$x\not\in A\cup B$. If membership in either $A$ or $B$ is decidable then clearly
$A$ and $B$ are not recursively inseparable, but the converse is not true. The reduction
in the previous proof always produced ap WFA whose $\omega$-function  is
either everywhere continuous, or not everywhere defined, so the recursive inseparability
follows directly.

Analogously, the proof shows that if {\sc Matrix Product Stability}
is undecidable then one cannot recursively separate  those finite matrix sets that
are continuously RCP  from those that are not RCP.
\end{remark}

Next we consider the implications if {\sc Matrix Product Stability} turns out to be decidable.

\begin{theorem}
\label{reduction2}
Problems {\sc Ap-WFA continuity} and
{\sc Matrix Product Continuity}
are algorithmically reducible to {\sc Matrix Product Stability}.
\end{theorem}

\begin{proof}
The reduction from {\sc Matrix Product Continuity} to {\sc Matrix Product Stability}
was proved in~\cite{Daubechies}. Let us prove the reduction from {\sc Ap-WFA continuity}, so
let $\A$ be a given ap automaton whose continuity we want to determine.

We begin by minimizing $\A$. If the resulting automaton computes the zero
function, we are done. Otherwise, we run the procedure from Lemma~\ref{matrix1}
to obtain the form (\ref{preform}) of transition matrices. If any step of the
algorithm fails (that is, nontrivial $I_c$ does not exist), $\A$ can not define a continuous
function. Otherwise, $\fa$ is continuous iff $\{B_a\,|\,a\in\Sigma\}$ in (\ref{preform})
is stable.
\end{proof}

From Theorems~\ref{reduction1} and \ref{reduction2} we conclude that decision problems
{\sc Matrix Product Stability}, {\sc Ap-WFA continuity} and
{\sc Matrix Product Continuity} are computationally equivalent.

If we drop the requirement that the WFA is ap, we can make the following observation:

\begin{theorem}
\label{reduction3}
{\sc Matrix Product Convergence} is algorithmically reducible to the problem of
determining if a given WFA is everywhere defined.
\end{theorem}

\begin{proof}
Use the same reduction as in the proof of Theorem~\ref{is-zero}.
\end{proof}

\section{Real functions defined by WFA}
\label{realfunctionsection}

Let $\Sigma=\{0,1\}$ be the binary alphabet and $\A$ a WFA over $\Sigma$. Then we
can use $\fa$ to define the real function
$\fah:[0,1)\rightarrow\mathbb{R}$ via the binary addressing scheme on the half-open interval $[0,1)$.
For $w\in\So$ denote by $num(w)$
the real number $$num(w)=\sum _{i=1} ^{\infty}w_i2^{-i}.$$
Let $\Omega=\Sow$. It is easy to see that by taking $num_{|\Omega}$, we obtain a one-to-one
correspondence between words of $\Omega$ and numbers in the interval $[0,1)$.
Denote $bin$ the inverse mapping to $num_{|\Omega}$, i.e.
\[\forall w\in\Omega,\,bin(x)=w \Longleftrightarrow num(w)=x.\]
We emphasize that the correspondence is between sets
$[0,1)$ and $\Omega$, not $[0,1)$ and $\So$.
A point $x\in[0,1)$ with a word presentation of the form $bin(x)=v0^{\omega}$
for some $v\in\Ss$ is called \emph{dyadic}. Points without such a presentation
are \emph{non-dyadic}.

Let $f$ be a (partial) function from $\So$ to $\mathbb R$. Then we define
the corresponding (partial) real function $\fh:[0,1)\rightarrow\mathbb{R}$ by:
\[\fh(x)=f(bin(x)).\]
As usual, if $f(bin(x))$ is not defined then $\fh(x)$ remains undefined.

\subsection{Continuity of real functions defined by WFA}
We will call the real function $\fh$ \emph{continuous} resp. \emph{uniformly
continuous} if it is continuous resp. uniformly continuous in the whole $[0,1)$. Note that $\fh$ is
uniformly continuous iff it can be extended to a continuous function on the
whole closed interval $[0,1]$.

The following two examples show that the function
$\fah$ can be continuous without being uniformly continuous: in these examples the left limit
$\ds\lim_{x\to1_{-}}\fah(x)$ does not exist.

\begin{example}\label{not-converge}
The ap WFA in Figure~\ref{not-converge-figure} computes a piecewise linear function
$\fah:[0,1)\rightarrow\mathbb{R}$ that
does not have the left limit at point $1$ (see its graph in Figure~\ref{n-c-graph}). The $\omega$-function
$\fa$ is everywhere defined, but  the convergence at point $1^\omega$ is not uniform. Note that
$\fa(1^\omega)=1/2$. Function
$\fa$ is continuous at all points except $1^\omega$.

\picture{Automaton from Example~\ref{not-converge}}{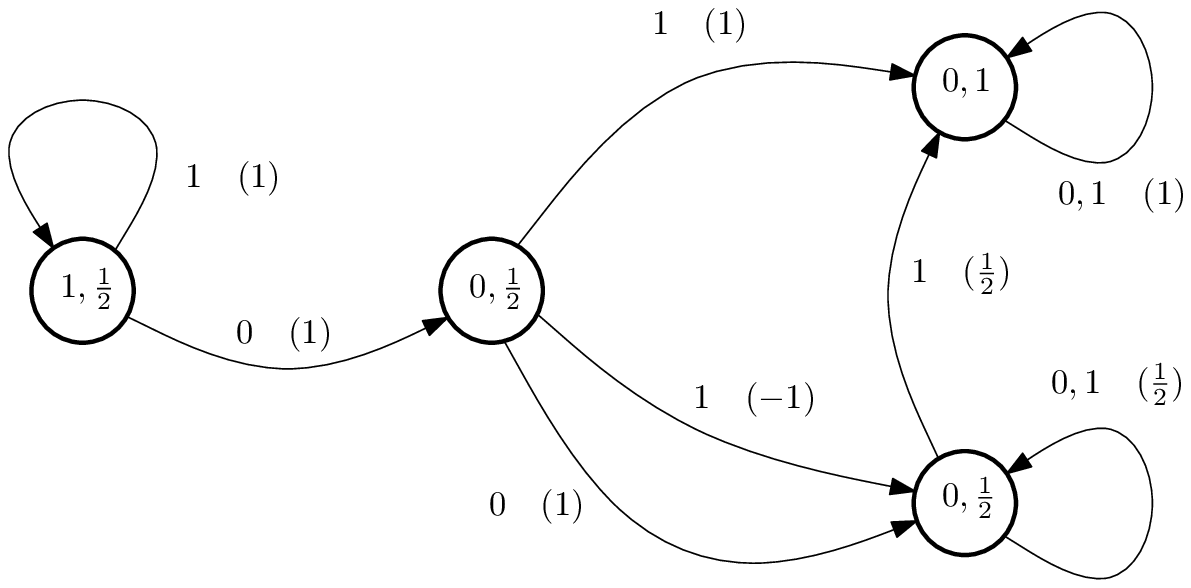}{not-converge-figure}

\picture{Graph from Example~\ref{not-converge}}{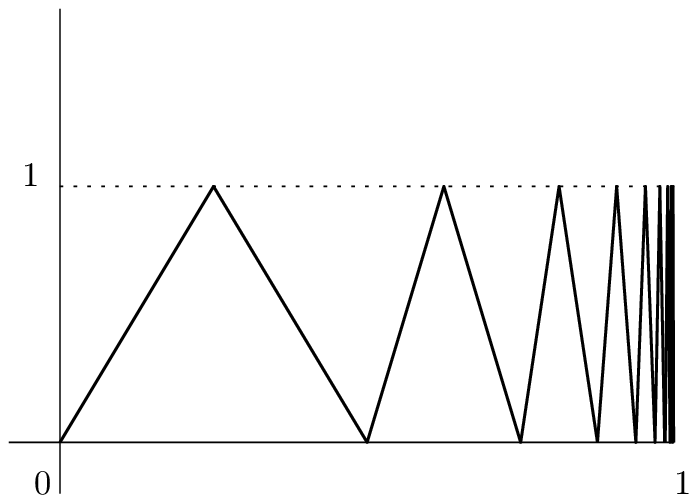}{n-c-graph}
\end{example}

\begin{example}\label{blow-up}
The ap WFA in Figure~\ref{blow-up-figure} computes a piecewise linear function that maps
$1-1/{2^n} \mapsto 2^n-1$ for $n\in\mathbb N$. See the graph in Figure~\ref{blow-up-graph}.
Obviously, $\ds\lim_{x\to1_{-}}f(x)=\infty$.  The $\omega$-function
$\fa$ is not defined at point $1^\omega$.
\picture{Automaton from Example~\ref{blow-up}}{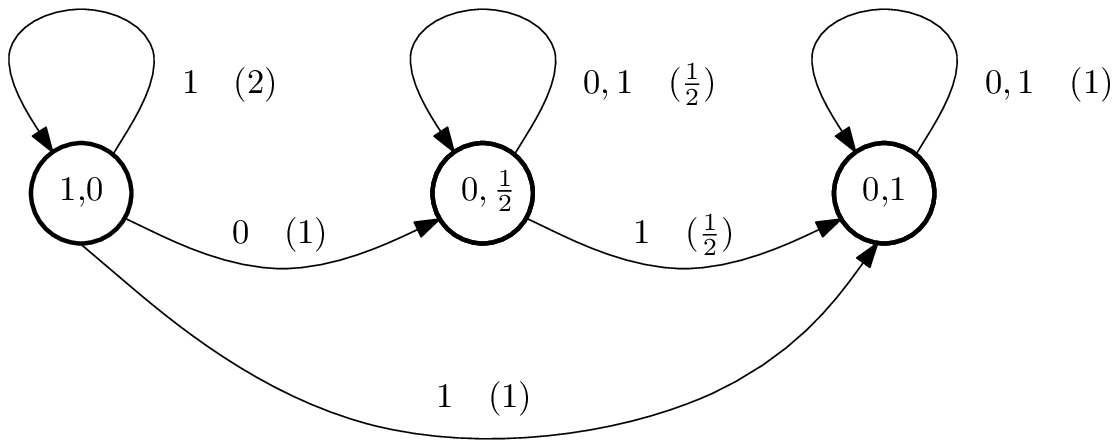}{blow-up-figure}
\picture{Graph from Example~\ref{blow-up}}{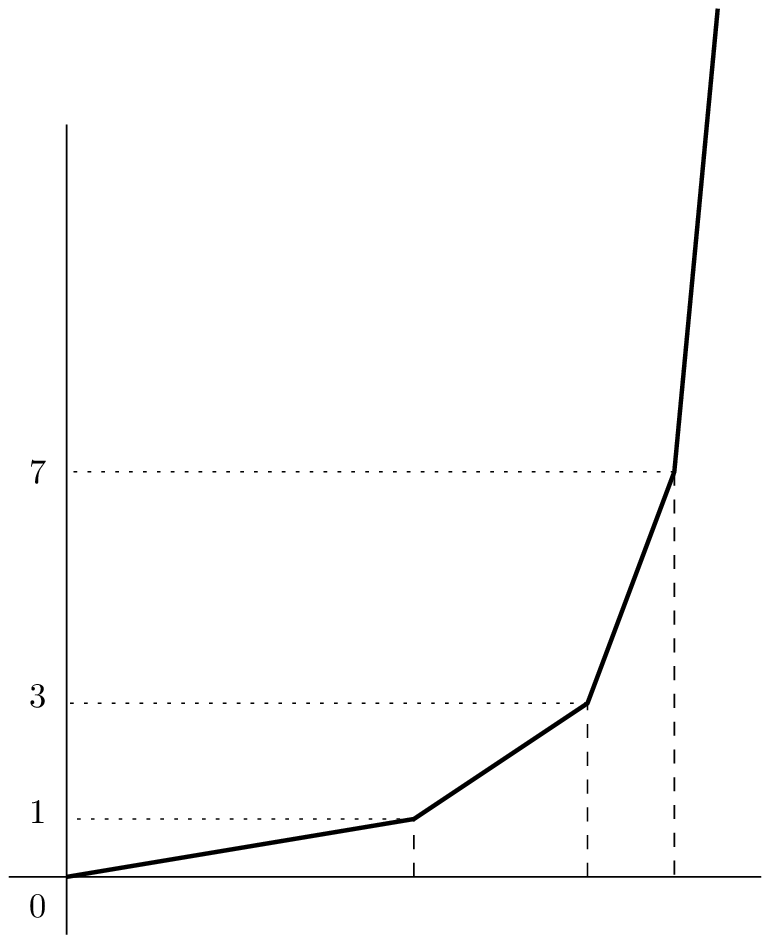}{blow-up-graph}
\end{example}

The following Lemma establishes correspondence between the continuity of the real function $\fh$
and the corresponding $\omega$-function $f$ in its relevant domain $\Omega$. Continuity of
$f$ in $\Omega$ corresponds to the
continuity of $\fh$ at all non-dyadic points together with continuity
of $\fh$ from the right at all dyadic points.

\begin{lemma}\label{realtowfa}
Let $f$ be any $\omega$-function, and let $\fh$ be the corresponding real function. Let
$x\in [0,1)$ and denote $w=bin(x)$.
Function $f$ is continuous at $w$ as a function $\Omega\to\mathbb{R}$
if and only if $\fh$ is continuous (continuous from the right) at the point $x$, provided
$x$ is non-dyadic (dyadic, respectively).
\end{lemma}
\begin{proof}
Let us show first that for $u,v\in \So$, we have the inequality between the Euclidean
and Cantor metrics
\[
d_E(num(u),num(v))\leq d_C(u,v).
\]
Let $d_C(u,v)= 2^{-j}$. Then $u_i=v_i\textrm{ for
all }1\leq i\leq j$. Therefore
\begin{align*}
d_E(num(u),num(v))&=\left|\sum _{i=1}^{\infty}(u_i-v_i)2^{-i}\right|
=\left|\sum_{i={j+1}}^{\infty}(u_i-v_i)2^{-i}\right|\\
&\leq\sum_{i={j+1}}^{\infty}\left|(u_i-v_i)\right|2^{-i}\leq 2^{-j}=d_C(u,v).
\end{align*}
We have obtained for all $u,v\in\So$ the implication
\[
d_C(u,v)<\delta\Longrightarrow d_E(num(u),num(v))<\delta,
\]
so it follows directly that the continuity of $\fh$ at $num(w)$ implies the continuity of $f$
at $w\in\Omega$. Suppose now that $num(w)$ is dyadic. Then $w=v0^\omega$ for some
finite word $v$ of length $k$. We have
\[
d_C(u,w)\leq 2^{-k} \Longrightarrow num(u)\geq num(w),
\]
so in this case continuity of $\fh$ at $num(w)$ from the right is enough to obtain the
continuity of $f$ at $w$.

Let us prove the converse direction. Suppose that $f$ is continuous at $w\in\Omega$.
For every $k$ there exists $\delta>0$ such that whenever
$num(w)\leq num(v) <num(w)+\delta$, then $pref_k(w)=pref_k(v)$. We can
accomplish this by choosing $\delta=num(pref_k(w)1^\omega)-num(w)$.

Similarly, if $w$ does not end in $0^\omega$ (i.e. $num(w)$ is not dyadic), we
can choose $\delta=num(w)-num(pref_k(w)0^\omega)$ and see that
$num(w)-\delta<num(v)\leq num(w)$ implies $pref_k(w)=pref_k(v)$.

This means that for every $\varepsilon>0$ there exists $\delta>0$ such that
$$
d_E(num(u),num(w))<\delta \mbox{ (and $num(u)>num(w)$ if $num(w)$ is dyadic) }
\Longrightarrow
d_C(u,w)<\epsilon.
$$
This is enough to see that $\fh$ is continuous at $x=num(w)$ if $x$ is not dyadic, and
continuous from right if $x$ is dyadic.
\end{proof}
The following example shows that Lemma~\ref{realtowfa} can not be extended to
continuity from the left at dyadic points.
\begin{example}
Let $\A$ be a WFA with \[I=\begin{pmatrix} 1 & 0 \end{pmatrix},\,
F=\begin{pmatrix} 0 \\ 1 \end{pmatrix},\,
A_0 =\begin{pmatrix} 0 & 0 \\ 0 & 1 \end{pmatrix}\textrm{ and }
A_1 = \begin{pmatrix} 0 & 1 \\ 0 & 1 \end{pmatrix}.\]
It is easy to see that $\fa(1v)=1$ and $\fa(0v)=0$ for all $v\in\So$.
Clearly, $\fa$ is continuous: For each $w,w'\in \So$,
$d_C(w,w')<1$ implies $d_E(\fa (w),\fa (w'))=0$.
However, $\fah$ is not continuous at the point $x=1/2$,
as $\fah (1/2)=1$, but $\fah (y)=0$  for any $y<1/2$.
\end{example}

Based on Lemma~\ref{realtowfa} we can now characterize those real functions
$\fh$ whose corresponding $\omega$-function $f$ is continuous or uniformly
continuous in $\Omega$.

\begin{corollary}\label{ST-function}
Let $f$ be an $\omega$-function and let $\fh$ be the corresponding
real function. Then:
\begin{enumerate}[(1)]
\item Function $f$ is continuous in the set $\Omega$ if and only if
$\fh$ is continuous at every non-dyadic point and continuous from the
right at every dyadic point.
\item Function $f$ is uniformly continuous in the set $\Omega$
if and only if $\fh$ is continuous at every non-dyadic point, continuous from the
right at every dyadic point, and has a limit from the left at all nonzero 
dyadic points as well as at the point $x=1$.
\end{enumerate}
Note that  $f$ might not even be defined at points in $\Ss1^\omega$.
\end{corollary}
\begin{proof}
Part (1) follows directly from Lemma~\ref{realtowfa}, so we focus on part (2).

Suppose that $f$ is uniformly continuous in $\Omega$.
By part (1) it is sufficient to show that $\fh$ has a limit from the
left at each point $num(v1^\omega)$ for $v\in\Ss$.
As $\Omega$ is dense in $\So$, there exists a (unique) continuous
$g:\So\to\mathbb R$ such that $g_{|\Omega}=f_{|\Omega}$. Then
\[
\lim_{x\to num(v1^\omega)_{-}} \fh(x)=
\lim_{w\to v1^\omega\above0pt w\in\Omega} f(w)
=
\lim_{w\to v1^\omega\above0pt w\in\Omega} g(w)
=g(v1^\omega),
\]
so the limit exists.

For the other direction of (2), assume that $\fh$
has a limit from the left at all dyadic points, including 1. By (1)
we have that $f$ is continuous in $\Omega$. We need to prove that $f$ is uniformly
continuous in $\Omega$. We show this by constructing a continuous
$g:\So\to\mathbb R$ such that $g_{|\Omega}=f_{|\Omega}$. Uniform continuity of
$f$ then
follows from the compactness of $\So$. For every $v1^\omega$, set
$$
g(v1^\omega)=
\lim _{x\rightarrow num(v1^\omega)_{-}} \fh(x)=\lim_{w\to
v1^\omega\above0pt w\in\Omega} f(w),
$$
while for $w\in\Omega$ we let $g(w)=f(w)$. Because the limit from the left exists at every
$num(v1^\omega)$, the function $g$ is everywhere defined.
It remains to verify that $g$ is continuous in $\So$.

Let $v\in\So$ and $\epsilon>0$. From the
definition of $g$ and properties of $\fh$ we obtain that there exists $\delta>0$ such that
\[
\forall u\in\Omega, d_C(v,u)<\delta \Rightarrow |g(v)-g(u)|<\frac12\epsilon.
\]
Now whenever $u=z1^{\omega}$ and $d_C(v,u)<\delta$, the value $g(u)$ is
the limit of the sequence $\{g(z1^n0^\omega)\}_{n=1}^\infty$ whose elements
belong to $\Omega$. Observe that for all $n$ large enough we have
$d_C(v,z1^n0^\omega)<\delta$ and so $|g(v)-g(z1^n0^\omega)|<\epsilon/2$. Therefore, $|g(v)-g(u)|<\epsilon$. 

We have shown for all $u$ that if
$d_C(v,u)<\delta$ then $|g(v)-g(u)|<\epsilon$, proving continuity.
\end{proof}

\begin{remark}
\label{uniformremark}
By (2) of Corollary~\ref{ST-function}, uniform continuity of $f$ in $\Omega$ implies the existence of
$\ds\lim _{x\rightarrow 1_-} \fh(x)$. So in this case, if $\fh$ is continuous it is uniformly
continuous. In particular, continuity of $f$ in $\So$ and $\fh$ in $[0,1)$ imply uniform
continuity of $\fh$.
\end{remark}

Uniform continuity of $\fh$ is stronger than uniform
continuity of $f$. The additional requirement is the continuity of $\fh$ from the left
at all dyadic points:

\begin{corollary}\label{cont}
The function $\fh:[0,1)\to{\mathbb R}$ obtained from the $\omega$-function $f$
is uniformly continuous if and only if:
\begin{enumerate}[(1)]
\item Function $f$ is uniformly continuous in $\Omega$, and
\item for all finite words $v$, the equality $g(v10^{\omega})=g(v01^{\omega})$ holds, where
$g$ is the (unique) continuous function $g:\So\to{\mathbb
R}$ such that $f_{|\Omega}=g_{|\Omega}$.
\end{enumerate}
\end{corollary}
\begin{proof}
If $\fh$ is uniformly continuous in $[0,1)$ then it has a right limit at $x=1$, so
$\fh$ satisfies the conditions in part (2) of Corollary~\ref{ST-function}.
Therefore, $f$ is uniformly continuous in $\Omega$. Let $g$ be the continuous extension of $f$ to $\So$.
Because $\gh$ is
continuous at dyadic points, we have
$$g(v10^{\omega})=\ds\lim_{w\to
v01^{\omega}}g(w)=g(v01^\omega).$$

Assume now that conditions (1) and (2) hold.

Using Lemma~\ref{realtowfa}, we obtain
continuity of $\fh$ at non-dyadic points and continuity from the right at
dyadic points. Now continuity of $\fh$ from the left at dyadic points follows
from (2) and the continuity of $g$.

We also have $\ds\lim_{x\to1_{-}}\fh(x)=g(1^\omega)$, so we can continuously
extend $\fh$ to the whole interval $[0,1]$, proving uniform continuity of $\fh$.
\end{proof}

If $\fah$ is uniformly continuous then we know that $f$ is uniformly continuous in $\Omega$.
Because $\Sigma^*0^\omega \subseteq \Omega$
we can choose $w=0^\omega$ and $\Delta=\Sigma^*0^\omega$ in
Theorem~\ref{redistribution} and obtain an average preserving WFA computing $f$. 
\begin{corollary}
\label{realuniformcorollary}
If a uniformly continuous function $\fah$
is computed by some WFA $\A$, then there is an
average preserving WFA $\B$ such that $\fah=\fbh$ and $\fb$ is continuous in $\So$.
Automaton $\B$ can be produced from $\A$ by first minimizing $\A$
and then changing the final distribution.
\end{corollary}

Note that $\B$ itself need not be right-minimal but we can minimize it.
Putting together the Corollary~\ref{realuniformcorollary} and Lemma~\ref{matrix1}, we obtain
the main result of this section:
\begin{corollary}
\label{corollary47}
If a nonzero uniformly continuous
function $\fh$ is computed by some WFA $\A$
then $\fh$ is also computed by a minimal, average preserving WFA with transition matrices of the form
\[A_i=\begin{pmatrix}
 B_i & \vline & \mathbf{b}_i \\ \hline
 \mathbf{0} & \vline & 1
\end{pmatrix},
\]
where $i=0,1$ and $\{B_0,B_1\}$ is a stable set of matrices.
\end{corollary}

\subsection{Decision problems concerning the real function continuity}

In this section we study how does the decision problem {\sc Matrix Product Stability} relate to the
problem of deciding the uniform continuity of the real function determined by a WFA.

Note that we do not address non-uniform continuity of $\fah$ for which
Corollary~\ref{corollary47} fails. On the other hand, by Corollary~\ref{realuniformcorollary}
any uniformly continuous $\fah$ is generated by an ap WFA with continuous $\fa$, so
we restrict the attention to such WFA. The decision problem of interest is then the following:
\bigskip

\noindent
{\sc Ap-WFA uniform continuity}:

\nobreak
{\bf Input:} An average preserving WFA $\A$ over the binary alphabet $\Sigma=\{0,1\}$.

\nobreak
{\bf Question:} Are both $\fa$ and $\fah$ everywhere continuous?
\bigskip

\noindent
Note that the question is equivalent to asking about the
uniform continuity of $\fa$ and $\fah$ (see Remark~\ref{uniformremark}).

To decide {\sc Ap-WFA uniform continuity} we need to verify that $\fa$ is
continuous and then check the condition (2) of Corollary~\ref{cont}. It turns
out that, if $\A$ is ap and $\fa$ continuous, condition (2) is easy to test.

\begin{lemma}\label{cond2}
Let $\A$ be an average preserving WFA such that $\fa$ is continuous on $\So$.
Then condition (2) of Corollary~\ref{cont} is decidable for the function $\fa$.
\end{lemma}
\begin{proof}
As minimization is effective we can assume that
the input automaton $\A$ is minimal and average preserving. First we can effectively
check whether $\fa=0$, in which case condition (2) of Corollary~\ref{cont} is
satisfied. Suppose than that $\fa\neq 0$. By Lemma~\ref{matrix1} we can
effectively transform the automaton to the form with transition matrices
\[
A_0=\begin{pmatrix}
B_0&\vline&\mathbf{b}_0\\\hline
{\mathbf 0}&\vline&1\\
\end{pmatrix},
\hspace*{1cm}
A_1=\begin{pmatrix}
B_1&\vline&\mathbf{b}_1\\\hline
{\mathbf 0}&\vline&1\\
\end{pmatrix},
\]
where $\{B_0, B_1\}$ is a stable set.
Because $\fa$ is continuous on $\So$  the condition (2) says that for all $v\in\Ss$
\[
\fa(v10^\omega)=\fa(v01^\omega).
\]
From minimality we obtain that the sufficient and necessary condition for this
to hold is that $A_{01^\omega}=A_{10^\omega}$.

Consider matrices $A_{01^k}$ and $A_{10^k}$. They are of the following forms:
\begin{align*}
  A_{01^k}&=\begin{pmatrix}
  B_0B_1^k & \vline & \mathbf{b}_0+B_0(\sum _{i=0}^{k-1}B^i_1)\mathbf{b}_1\\ \hline
  \mathbf{0} & \vline & 1
\end{pmatrix}\\
A_{10^k}&=\begin{pmatrix}
  B_1B_0^k & \vline & \mathbf{b}_1+B_1(\sum _{i=0}^{k-1}B^i_0)\mathbf{b}_0\\ \hline
  \mathbf{0} & \vline & 1
\end{pmatrix}.\\
\end{align*}

Observe that
\[
\sum_{i=0}^{k-1}B^i_0(E-B_0)=E-B^k_0\;\textrm{ and }
\sum _{i=0}^{k-1}B^i_1(E-B_1)=E-B^k_1.
\]

As the set $\{B_0,B_1\}$ is stable, we must have $B_0^n,\,B_1^n\to 0$ and so
all eigenvalues of both $B_0$ and $B_1$ must lie inside the unit disc. Thus the sums
$\sum_{i=0}^\infty B_0^i$ and $\sum_{i=0}^\infty B_1^i$ converge. It follows that
\[
\sum _{i=0}^{\infty}B^i_0=(E-B_0)^{-1}\quad\textrm{ and }\quad
\sum _{i=0}^{\infty}B^i_1=(E-B_1)^{-1}.
\]
This means that we have the limits:
\begin{align*}
  A_{01^\omega}&=\begin{pmatrix}
  0 & \vline & \mathbf{b}_0+B_0(E-B_1)^{-1}\mathbf{b}_1\\ \hline
  \mathbf{0} & \vline & 1
\end{pmatrix}\quad \textrm{and}\\
A_{10^\omega}&=\begin{pmatrix}
0 & \vline & \mathbf{b}_1+B_1(E-B_0)^{-1}\mathbf{b}_0\\ \hline
  \mathbf{0} & \vline & 1
\end{pmatrix}.
\end{align*}

So we are left with the simple task of checking the equality
$$
\mathbf{b}_0+B_0(E-B_1)^{-1}\mathbf{b}_1 = \mathbf{b}_1+B_1(E-B_0)^{-1}\mathbf{b}_0.
$$
\end{proof}

We are ready to prove the main result of this section. Recall that for $\A$
average preserving, continuity of $\fa$ is computationally as hard as
stability. We show that also simultaneous continuity of $\fa$ and $\fah$ is as
hard.

\begin{theorem}\label{equiv}
Decision problems {\sc Matrix Product Stability} and {\sc Ap-WFA uniform continuity}
can be algorithmically reduced to each other.
\end{theorem}

\begin{proof}
Suppose first that {\sc Matrix Product Stability} is decidable, and let $\A$ be a given
ap WFA over the binary alphabet. By Theorem~\ref{reduction2} we can effectively
determine if $\fa$ is continuous in $\So$. If the answer is positive then -- according to Lemma~\ref{cond2} --
we can effectively check whether the function $\fa$ satisfies the condition (2) in Lemma~\ref{cont}.
By Lemma~\ref{cont} this is enough to determine whether $\fah$ is uniformly continuous, so we
get the answer to {\sc Ap-WFA uniform continuity}.

For the converse direction, let us assume that {\sc Ap-WFA uniform continuity} is decidable.
By Lemma~\ref{reductiontosizetwo} it is enough to show how we can determine if a given pair
$\{B_0, B_1\}$ of $n\times n$ matrices is stable. Because we can
check whether $\ds\lim_{n\to\infty} B_i^n=0$ for $i=0,1$ (using the Lyapunov equation
method as in \cite[page 169]{Mahmoud}), we can assume that $\{B_0\}$ and
$\{B_1\}$ are stable sets.

In the following we effectively construct several  ap WFA $\A_{ij}$
over the binary alphabet such that
\begin{itemize}
\item if $\{B_0, B_1\}$  is stable then  the functions ${\fa}_{ij}$ and
${{\hat{f}}_{{\mathcal{A}}_{ij}}}$
are continuous for each $i,j$, while
\item if $\{B_0, B_1\}$  is not stable then for some $i,j$ the function ${\fa}_{ij}$ is
not continuous.
\end{itemize}
The result then follows directly.
The construction of $\A_{ij}$ is similar to the proof of
Theorem~\ref{reduction1}. Again, we write down the transition matrices in the
the block form
\[
A_0=
\begin{pmatrix}
 B_0 & \vline & C_0 \\
 \hline
 \mathbf{0} & \vline & D_0
\end{pmatrix},\quad
A_1=
\begin{pmatrix}
 B_1 & \vline & C_1 \\
 \hline
 \mathbf{0} & \vline & D_1
\end{pmatrix},
\]
only this time, instead of constant $D_0=D_1=1$, we use the $3\times 3$ matrices
\[
D_0=\begin{pmatrix}
0& 1& 0\\
0& \frac12& 0 \\
0& 0& 1\\
\end{pmatrix},
\quad
D_1=\begin{pmatrix}
0& -1& 1\\
0& \frac12& \frac12\\
0& 0& 1\\
\end{pmatrix}.
\]
These matrices (with initial and final distributions
$I=(1,0,0)$ and $F=(1/2,1/2,1)^T$)
form a minimal ap WFA $\D$ that computes
the continuous real function shown in Figure~\ref{broken-line}. An important
feature of this function, implicit in the proof below, is the fact that it has value
zero at both endpoints of the domain interval. Also, by Theorem~\ref{WFA-RCP}, $\{D_0, D_1\}$
is a continuous RCP set.
\picture{The graph of the automaton $\D$}{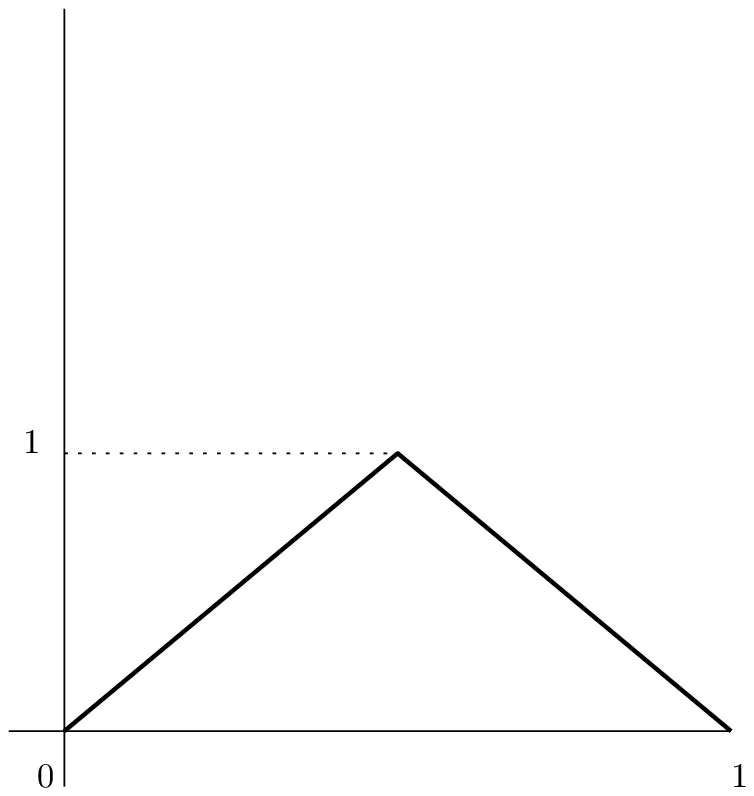}{broken-line}

Let $i,j\in\{1,\dots n\}$.
Denote by $C_0$ the following $n\times 3$ matrix:
\[
C_0=\begin{pmatrix}
0 &0&0\\
&\vdots&\\
0 &0&0\\
1 &0&0\\
0 &0&0\\
&\vdots&\\
0&0&0\\
\end{pmatrix},
\]
where the single 1 is in the $j$-th row. Let $C_1=-C_0$.

We now construct the ap WFA $\A_{ij}$ with transition matrices
\[
A_0=
\begin{pmatrix}
 B_0 & \vline & C_0 \\
 \hline
 \mathbf{0} & \vline & D_0
\end{pmatrix},
\hspace*{1cm}
A_1=
\begin{pmatrix}
 B_1 & \vline & C_1 \\
 \hline
 \mathbf{0} & \vline & D_1
\end{pmatrix},
\]
 initial distribution $I_i$ ($i$-th element of the canonical basis)
and final distribution $F=(0,\dots,0,1/2,1/2,1)^T$. 

Assume for a moment that $\faij$ is continuous. We show that then
\begin{eqnarray*}
G &=& \lim_{n\to\infty}A_0^nF = (0,\dots,0,0,1)^T,\\
H &=& \lim_{n\to\infty}A_1^nF = (0,\dots,0,1,1)^T.
\end{eqnarray*}
Consider only $G$; the case of $H$ is similar.
As $\{B_0\}$ is stable and $\{D_0\}$ is RCP,
an application of Lemma~\ref{matrix2} on the singleton set $\{A_0\}$
shows that the limit $G$ exists.
The vector $G$ is a 1-eigenvector of $A_0$ and by direct computation
we obtain that the last three elements of $G$ are $0,0,1$.

Notice now that the vector $G'=(0,\dots,0,1)^T$ is a 1-eigenvector of $A_0$. Were
$G\neq G'$, we would have the 1-eigenvector $G-G'$ whose last three elements
are zero. But then the first $n$ elements of $G-G'$ form a 1-eigenvector of
$B_0$ and so $\{B_0\}$ is not stable, a contradiction. Thus $G=G'$.
The proof that $H=(0,\dots,0,1,1)^T$ is analogous.

We are now ready to finish the proof.
Assume first that $\{B_0,B_1\}$ is a stable set. We claim that then
$\faij,\faijh$
are both continuous. 
Now the general form of Lemma~\ref{matrix2} comes into play: According to part
(3) of that Lemma, the set $\{A_0,A_1\}$ is a continuous RCP set and so, by Theorem~\ref{WFA-RCP},
the function $\faij$ is continuous.

By Corollary~\ref{cont}, we only need to show that condition (2) of that Corollary is satisfied.
We can compute the limits
\[
\lim_{n\to\infty}A_1A_0^nF=A_1G=\begin{pmatrix}
0\\
\vdots\\
0\\
1\\
\frac12\\
1\\
\end{pmatrix}=A_0H=\lim_{n\to\infty}A_0A_1^nF.
\]
This implies that $\faij(v01^\omega)=\faij(v10^\omega)$ for all $v\in\Sigma^*$, so by
Corollary~\ref{cont} the function $\faijh$ is continuous.

Suppose then that the set $\{B_0,B_1\}$ is not stable. Then there exist $i,j$ and
$w\in\So$ such that for some $\epsilon>0$ there are infinitely many $n$ such
that $|(B_{pref_n(w)})_{i,j}|>\epsilon$. Consider the automaton $\A_{ij}$ for these $i,j$.
We want to prove that $\faij$ is not continuous in this case.

We will proceed by contradiction, assuming that $\faij$ is continuous. Then
$\A_{ij}$
is uniformly convergent by Lemma~\ref{sup}. Then from Lemma~\ref{apu} and
continuity of $\faij$ we obtain that
\[
\lim_{n\to\infty}\left[2\Faij(pref_n(w)0)-\faij(pref_n(w)10^\omega)-\faij(pref_n(w)0^\omega)\right]=0.
\]

This means that $\ds\lim_{n\to\infty}IA_{pref_n(w)}(2A_0F-A_1G-G)=0$. However,
a straightforward calculation shows that the
vector $2A_0F-A_1G-G$ is the $j$-th element of the canonical basis and so
$IA_{pref_n(w)}(2A_0F-A_1G-G)=(A_{pref_n(w)})_{i,j}$ which does not converge to
zero. Therefore, $\faij$ can not be continuous.
\end{proof}

\begin{remark}
Analogously to Remark~\ref{inseparable} we can note that
in the case that {\sc Matrix Product Stability} is undecidable we have in fact showed
the recursive inseparability
of ap WFA whose $\fa$ and $\fah$ are both continuous from those ap WFA whose $\fa$ is not continuous.
\end{remark}

\subsection{Constructing WFA defining continuous real functions}
We end our paper by giving a few notes on how to construct nontrivial ap
WFA with continuous $\fa$ and $\fah$, for all initial distributions.

\begin{lemma}\label{constant}
Let $\A$ be a left-minimal ap automaton. Then the following statements
are equivalent:
\begin{enumerate}[(1)]
\item $\fa$ is constant
\item $\Fa$ is constant
\item $A_aF=F$ for all $a\in\Sigma$.
\end{enumerate}
\end{lemma}
\begin{proof}
Implications $(3)\Rightarrow (2)\Rightarrow (1)$ are obvious. We prove
$(1)\Rightarrow (2)$ and $(2)\Rightarrow (3)$.

Assume (1):
$\fa(w)=c$ for all $w\in\So$. Then $\fa$ is the $\omega$-function
corresponding to both $\Fa$ and the constant word
function $G(u)=c$. Then $\Fa-G$ is an average preserving
word function whose $\omega$-function is zero, so by Lemma~\ref{zero-ap-lemma} we have
$\Fa-G=0$, i.e. condition (2) holds.

To prove (3), assuming (2), we note that
for all $u\in\Sigma^*$ and $a\in\Sigma$ the equality
\[
IA_uA_aF=\Fa(ua)=\Fa(u)=IA_uF
\]
holds. By left minimality this implies $A_aF=F$.
\end{proof}

Notice that even without left-minimality we have the following:
if $\A$ is an
ap WFA such that $A_aF\neq F$ for some $a\in \Sigma$
then there exists a choice for the initial distribution $I$ such that $\fa$ is not constant.

\begin{lemma}\label{regular}
Let $\{B_0,B_1\}$ be a stable set of matrices. Then
$\det(B_0+B_1-2E)\neq 0$.
\end{lemma}
\begin{proof}
Were it not the case, there would exist a vector $v\neq 0$ such that for each
$n$ we would have
\[
\left(\frac{B_0+B_1}2\right)^n v=v.
\]

Then we can write:
\[
\left\|v\right\|=\left\|\left(\frac{B_0+B_1}2\right)^nv\right\|=\left\|\sum_{w\in
\Sigma^n}\frac{B_wv}{2^n}\right\|\leq \sum_{w\in
\Sigma^n}\frac{\left\|B_w\right\|}{2^n}\left\|v\right\|.
\]
However, by Lemma~\ref{uniform} there
exists $n$ such that $\left\|B_w\right\|<1$ for each $w$ of length $n$.
For such $n$,
\[
\left\|v\right\|
< \sum_{w\in\Sigma^n}\frac{1}{2^{n}}\left\|v\right\|=\left\|v\right\|,
\]
a contradiction.
\end{proof}

The following theorem (and its proof) gives us tools to generate
ap WFA with non-constant, continuous $\fa$ and $\fah$. We know
from Corollary~\ref{corollary47} that we can limit the search to ap WFA with
transition matrices in the form (\ref{preform}), for stable $\{B_0,B_1\}$.
The minimality condition can be replaced by the weaker concept that
all initial distribution yield a continuous WFA.

\begin{theorem}\label{determinator}
Let $\{B_0,B_1\}$ be a stable set of matrices. Consider the problem of finding
vectors $\mathbf{b}_0,\mathbf{b}_1$ and a final distribution $F$
so that, for any choice of the initial distribution $I$, the transition matrices
\[
A_i=
\begin{pmatrix}
 B_i & \vline & \mathbf{b}_i \\ \hline
  \mathbf{0} & \vline & 1
\end{pmatrix},\,i=0,1,
\]
describe an ap WFA $\A$ with continuous $\fa$ and $\fah$. We also want $A_0F\neq F$,
so that for some initial distribution $\A$ does not define the constant function.
\begin{enumerate}[(1)]
\item If $\det(B_0+B_1-E)=0$ then
we can algorithmically find such vectors $\mathbf{b}_0$, $\mathbf{b}_1$ and $F$.
\item If $\det(B_0+B_1-E)\neq 0$ then such choices do not exist: only the constant function
$\fa$ can be obtained.
\end{enumerate}
\end{theorem}
\begin{proof}
We are going to obtain sufficient and necessary
conditions for the vectors $\mathbf{b}_0,\mathbf{b}_1$ and $F$.

By definition, the ap condition is $(A_0+A_1)F= 2F$. Let $F'$ be the vector
obtained from $F$ by removing the last element. Note that the last element of
$F$ cannot be zero, because then the ap condition would require $(B_0+B_1)F'=
2F'$, which only has the solution $F'=0$ by Lemma~\ref{regular}. Without loss
of generality, we fix the last element of $F$ to be 1. (We can do this because
multiplication of the final distribution by any non-zero constant $c$ has only
the effect of multiplying $\fa$ and $\fah$ by $c$.)

The ap condition becomes
\[
\begin{pmatrix}
 B_0+B_1 & \vline & \mathbf{b}_0+\mathbf{b}_1 \\ \hline
  \mathbf{0} & \vline & 2
\end{pmatrix}
\begin{pmatrix}
F'\\
1
\end{pmatrix}
=2\begin{pmatrix}
F'\\
1
\end{pmatrix},
\]
that is,
\begin{equation}
\label{apeq}
(B_0+B_1-2E)F'+\mathbf{b}_0+\mathbf{b}_1=0.
\end{equation}
From Lemma~\ref{regular}, we have that $B_0+B_1-2E$ is regular. This means
that for any choice of vectors $\mathbf{b}_0,\mathbf{b}_1$
there is a unique $F'$, given by (\ref{apeq}), that makes the WFA average preserving.

The requirement that $\fa$ is continuous is automatically satisfied as $\{B_0, B_1\}$ is stable
(the case (3) of Lemma~\ref{matrix2} and the case (2) of Theorem~\ref{WFA-RCP}).
By Corollary~\ref{cont} continuity of $\fah$ is then equivalent to
the condition $\fa(v10^{\omega})=\fa(v01^{\omega})$ for all $v\in\Sigma^*$.
Since we require $\fah$ to be continuous for all initial distributions, we have
the equivalent condition that
\[
\ds\lim
_{k\rightarrow\infty}
(A_0A_1^k)F=\ds\lim_{k\rightarrow\infty}(A_1A_0^k)F.
\]
As in the proof of Lemma~\ref{cond2}, we obtain

\begin{align*}
A_{01^\omega}&=\begin{pmatrix}
  0 & \vline & \mathbf{b}_0+B_0(E-B_1)^{-1}\mathbf{b}_1\\ \hline
  \mathbf{0} & \vline & 1
\end{pmatrix}
\quad \textrm{and}\\
A_{10^\omega}&=\begin{pmatrix}
0 & \vline & \mathbf{b}_1+B_1(E-B_0)^{-1}\mathbf{b}_0\\ \hline
  \mathbf{0} & \vline & 1
\end{pmatrix}
\end{align*}
Therefore, we can rewrite $A_{01^\omega}F=A_{10^\omega}F$ as an equation for
vectors $\mathbf{b}_0$ and $\mathbf{b}_1$:
\[
 \mathbf{b}_0+B_0(E-B_1)^{-1}\mathbf{b}_1=\mathbf{b}_1+B_1(E-B_0)^{-1}\mathbf{b}_0
.\]
This can be written equivalently as
\begin{equation}
\label{conteq}
(B_0+B_1-E)[(E-B_1)^{-1}\mathbf{b}_1-(E-B_0)^{-1}\mathbf{b}_0] = 0
\end{equation}
So choices of $\mathbf{b}_0$, $\mathbf{b}_1$ and $F$ that satisfy the requirements
of the theorem (except for $A_0F\neq F$) are exactly the ones that satisfy
(\ref{apeq}) and (\ref{conteq}).

Consider now the final requirement $A_0F\neq F$. This is equivalent to
\[
(B_0-E)F'+\mathbf{b}_0\neq 0,
\]
and further to
$F' \neq -(B_0-E)^{-1}\mathbf{b}_0$.
Substituting for $F'$ in the ap condition (\ref{apeq}), and recalling that matrix $B_0+B_1-2E$ is
regular, we obtain the equivalent condition
$$
-\mathbf{b}_0- (B_1-E)(B_0-E)^{-1}\mathbf{b}_0+\mathbf{b}_0+\mathbf{b}_1 \neq 0,
$$
which can be rewritten as
\begin{equation}
\label{nonconstanteq}
(E-B_1)^{-1}\mathbf{b}_1-(E-B_0)^{-1}\mathbf{b}_0 \neq 0.
\end{equation}
We have obtained  sufficient and necessary conditions (\ref{apeq}), (\ref{conteq}) and
(\ref{nonconstanteq}).

Now we can prove parts (1) and (2) of the theorem.
If $\det(B_0+B_1-E)\neq 0$ then (\ref{conteq}) and (\ref{nonconstanteq}) are contradictory,
so no choice of $\mathbf{b}_0$, $\mathbf{b}_1$ and $F$ can satisfy all the requirements.
On the other hand, if $\det(B_0+B_1-E) = 0$ we can choose
$\mathbf{b}_0,\mathbf{b}_1$ so that $(E-B_1)^{-1}\mathbf{b}_1-(E-B_0)^{-1}\mathbf{b}_0$ is a
nonzero element of the kernel of matrix $B_0+B_1-E$. This can be easily done by, for example,
choosing any nonzero $\mathbf{k}\in ker(B_0+B_1-E)$ and an arbitrary vector
$\mathbf{b}_0$, and setting
$$
\mathbf{b}_1 = (E-B_1) \left[\mathbf{k} + (E-B_0)^{-1}\mathbf{b}_0\right].
$$
These choices of $\mathbf{b}_0$ and $\mathbf{b}_1$ satisfy (\ref{conteq}) and (\ref{nonconstanteq}).
We can then calculate the unique $F'$ that satisfies (\ref{apeq}).
\end{proof}

We see that in order to generate non-constant functions we need a
stable pair of matrices $\{B_0,B_1\}$ such that
$\det(B_0+B_1-E)=0$.

The following numerical example illustrates the previous proof.
\begin{example}
\label{contex}
Let
$$
B_0=
\left(
\begin{array}{rr}
\frac{1}{3} & \frac{1}{3} \\
\frac{1}{3} & \frac{1}{3}
\end{array}\right), \hspace*{1cm}
B_1=
\left(
\begin{array}{rr}
\frac{2}{3} & 0 \\
-\frac{1}{3} & \frac{2}{3}
\end{array}\right).
$$
It is easy to see that $\{B_0, B_1\}$ is stable and $det(B_0+B_1-E)=0$. The kernel of $B_0+B_1-E$
is generated by $(1,0)^T$. If we (arbitrarily)
choose $\mathbf{k}=(9,0)^T$ and $\mathbf{b}_0=(3,0)^T$ we can solve
$$
\mathbf{b}_1 = (E-B_1) \left[\mathbf{k} + (E-B_0)^{-1}\mathbf{b}_0\right] = (5,6)^T.
$$
From (\ref{apeq}) we get
$$
F' = -(B_0+B_1-2E)^{-1}(\mathbf{b}_0+\mathbf{b}_1) = (10,6)^T.
$$
So we have the ap WFA
$$
A_0=
\left(
\begin{array}{rrr}
\frac{1}{3} & \frac{1}{3} &3 \\
\frac{1}{3} & \frac{1}{3} &0\\
0 & 0 & 1
\end{array}\right), \hspace*{1cm}
A_1=
\left(
\begin{array}{rrr}
\frac{2}{3} & 0 & 5\\
-\frac{1}{3} & \frac{2}{3} &6\\
0 & 0& 1
\end{array}\right),\hspace*{1cm}
F=
\left(
\begin{array}{c} 10 \\ 6 \\ 1\end{array}
\right)
$$
which with the initial distribution $(1,0,0)$ defines the real function $\fah$ whose graph is shown in Figure~\ref{contfig}.

\picture{The continuous function specified by the ap WFA in
Example~\ref{contex}}{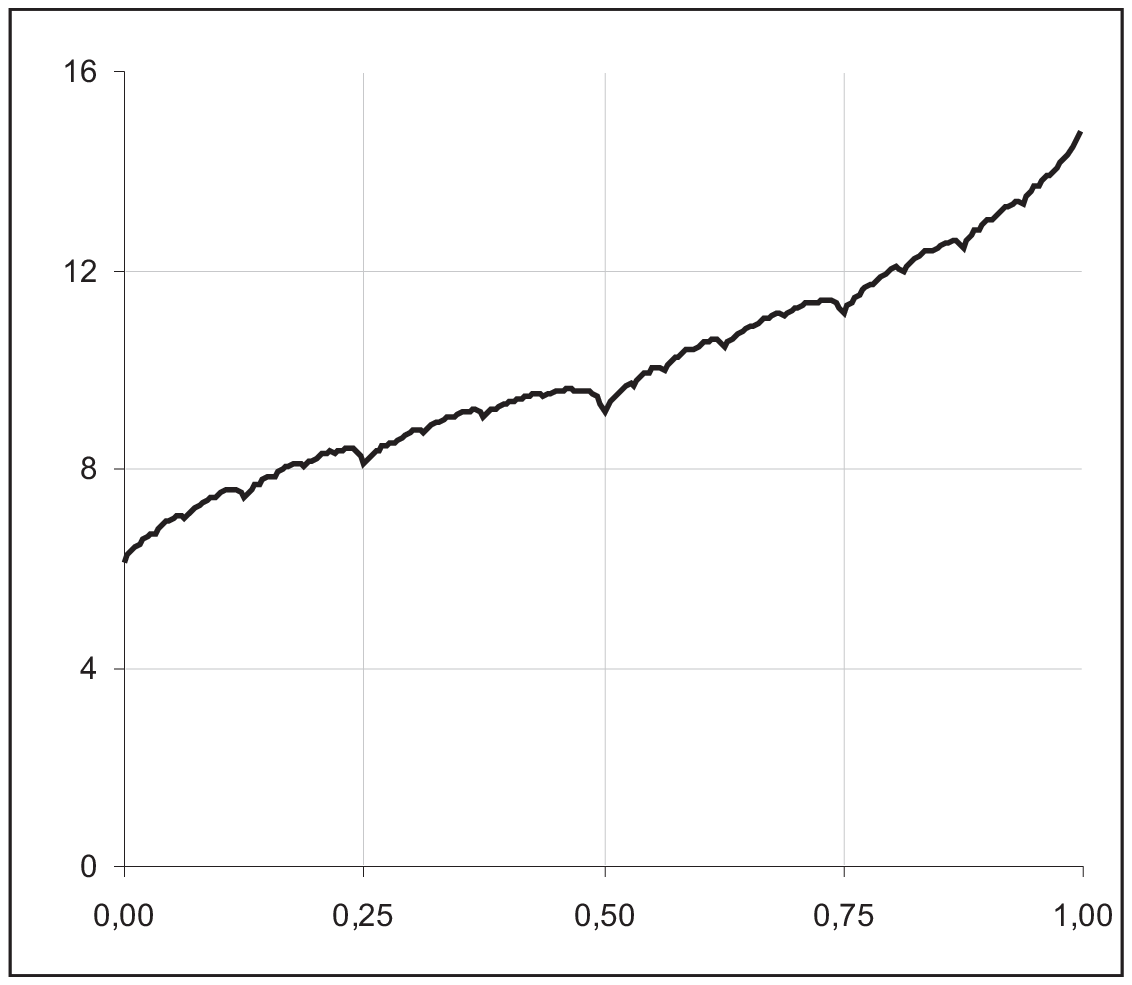}{contfig}

\end{example}

\begin{example}
It is easy to see that one-state continuous ap WFA can compute precisely all
constant functions.

Let us find all two-state ap WFA with continuous $\fa$ and $\fah$.
Now $B_0$ and $B_1$ are numbers, and the condition that $det(B_0+B_1-E)=0$ forces them to add up to one.
Stability requires both numbers to be less than 1 in absolute value, so we have $B_0=a$ and $B_1=1-a$
for some $0<a<1$. We can choose $b_0$ and $b_1$ arbitrarily, and calculate $F'=b_0+b_1$.
We get the continuous ap WFA with
$$
A_0=
\left(
\begin{array}{cc}
a & b_0 \\
0 & 1
\end{array}\right), \hspace*{1cm}
A_1=
\left(
\begin{array}{cc}
1-a & b_1\\
0 & 1
\end{array}\right),\hspace*{1cm}
F=
\left(
\begin{array}{c} b_0+b_1 \\ 1\end{array}
\right),
$$
for $0<a<1$ and $b_0, b_1\in \mathbb{R}$.
Note that we did not require (\ref{nonconstanteq}) to hold, which means that we also get
the constant functions when
$$
\frac{b_1}{a} = \frac{b_0}{1-a}.
$$
\end{example}

\section{Conclusions}
We have investigated the relationship between continuity of WFA and properties
of its transition matrices. We have obtained a
``canonical form''  for
ap WFA computing continuous functions
(the form (\ref{preform}) from Lemma~\ref{matrix1}). These
results generalize some of the theorems in \cite{Karh} and are
similar to those obtained in a slightly different setting in the
article \cite{Daubechies}. Moreover, we present a method of constructing
continuous WFA.

We have also asked questions about decidability of various incarnations of the
continuity problem. Mostly, these problems turn out to be equivalent to the
Matrix Product Stability problem. This is why we believe that any interesting
question about continuity of functions computed by WFA is at least as hard as
Matrix Product Stability.

There are numerous open questions in this area. Most obviously, settling
the decidability of the Matrix Product Stability problem would be a great
step forward. However, as this problem has resisted efforts of
mathematicians so far, we offer a few other open problems:

\begin{question}
Given an automaton computing a continuous $\omega$-function, can we
algorithmically find the ap automaton computing the same function?
\end{question}

\begin{question}
Given ap automaton computing $\omega$-function which is uniformly continuous on
$\Omega$, can we algorithmically find automaton computing the function $g$ from
Theorem~\ref{redistribution}?
\end{question}
\begin{question}
Is deciding the continuity of $\fah$ for ap automata computationally equivalent
with deciding Matrix Product Stability?
\end{question}

Other interesting  questions that can be posed on WFA are whether
a given $\fah$ converges everywhere, and whether it is bounded. We
know that all level WFA (as described in \cite{Karh}) are both everywhere convergent and bounded
but both properties remain to be characterized in the
general case. We also point out that similar results on higher differentiability classes
(e.g. continuously differentiable WFA functions) are likely to exist and should be
investigated.

\bibliographystyle{elsarticle-num}
\bibliography{citations}
\end{document}